\documentclass[sn-mathphys-num]{sn-jnl}

\usepackage{graphicx}%
\usepackage{multirow}%
\usepackage{amsmath,amssymb,amsfonts}%
\usepackage{amsthm}%
\usepackage{mathrsfs}%
\usepackage[title]{appendix}%
\usepackage{xcolor}%
\usepackage{textcomp}%
\usepackage{manyfoot}%
\usepackage{booktabs}%
\usepackage{algorithm}%
\usepackage{algorithmicx}%
\usepackage{algpseudocode}%
\usepackage{listings}%
\usepackage{epstopdf}
\usepackage{mathtools}
\usepackage{bm}
\usepackage{dsfont}
\usepackage{cleveref}
\usepackage{enumitem}

\newcommand{\opdiv}{\operatorname{div}}
\newcommand{\T}{\mathsf{T}}
\DeclareMathOperator*{\argmin}{\textrm{argmin}}
\newcommand{\R}{\mathbb{R}}
\newcommand{\spanV}{\operatorname{span}}
\DeclareMathOperator*{\argmax}{\textrm{argmax}}
\DeclareMathOperator{\diag}{diag}
\newfloat{algorithm}{t}{lop}
\floatname{algorithm}{Algorithm}
\ifpdf
  \DeclareGraphicsExtensions{.eps,.pdf,.png,.jpg}
\else
  \DeclareGraphicsExtensions{.eps}
\fi

\newtheorem{theorem}{Theorem}
\newtheorem{proposition}[theorem]{Proposition}%

\begin{document}

\title[Connecting Image Inpainting with Denoising\\ in the Homogeneous Diffusion Setting]{Connecting Image Inpainting with Denoising\\ in the Homogeneous Diffusion Setting}


\author*[1]{\fnm{Daniel} \sur{Gaa}}\email{gaa@mia.uni-saarland.de}
\author[1]{\fnm{Vassillen} \sur{Chizhov}}\email{chizhov@mia.uni-saarland.de}
\author[1]{\fnm{Pascal} \sur{Peter}}\email{peter@mia.uni-saarland.de}
\author[1]{\fnm{Joachim} \sur{Weickert}}\email{weickert@mia.uni-saarland.de}
\author[1]{\fnm{Robin Dirk} \sur{Adam}}\email{adam@mia.uni-saarland.de}

\affil[1]{\orgdiv{Mathematical Image Analysis Group, Faculty of Mathematics and Computer Science}, \orgname{Saarland University}, \orgaddress{\street{Campus E1.7}, \postcode{66041} \city{Saarbr\"ucken}, \country{Germany}}}


\abstract{
While local methods for image denoising and inpainting may use similar
concepts, their connections have hardly been investigated so far.
The goal of this work is to establish links between the two by
focusing on the most foundational scenario on both sides:
the homogeneous diffusion setting.
To this end, we study a denoising by inpainting (DbI) framework: It 
averages multiple inpainting results from different noisy subsets.
We derive equivalence results between DbI on shifted regular grids and 
homogeneous diffusion filtering in 1D via an explicit relation 
between the density and the diffusion time. 
We also provide an empirical extension to the 2-D case. 
We present experiments that confirm our theory and suggest that it 
can also be generalized to diffusions with non-homogeneous data or
non-homogeneous diffusivities. More generally, our work demonstrates 
that the hardly explored idea of data adaptivity deserves more 
attention: It can be as powerful as some popular models with 
operator adaptivity.}




\maketitle


\section{Introduction}
\label{sec:introduction}

Investigating connections between different fields in image analysis 
has often been rewarded with deep structural insights.
Consider for example the link between variational image 
inpainting~\cite{BSCB00, EL99b, GL14, MM98a, Sc15} and optic flow 
computation~\cite{HS81, NE86, WS01} via the concept of the 
{\em filling-in effect}.
This effect is due to the smoothness term (regularizer) of the models,
which inserts information at locations where the data term is absent or small 
in magnitude. The gradient flow for minimizing the variational energy 
functional leads to partial differential equations (PDEs) with a 
diffusion term.

While the filling-in effect has an obvious benefit for image inpainting, 
it can also lead to more powerful optic flow methods.
It produces a dense flow field from the sparse information of the data term.
Surprisingly, the parts of the flow field that are filled in by the 
diffusion-like regularization terms are usually the ones with the highest 
confidence~\cite{BW05}.

\Cref{fig:introduction} shows a similar but hitherto hardly studied 
effect when performing \emph{sparse inpainting} on noisy data. There 
the known data -- 
the so-called \emph{mask} -- is a scattered set of pixels. The 
noisy mask pixels remain unchanged during the process, 
while the unknown areas in between are interpolated smoothly by averaging 
information from the noisy pixels. 
We thus again have a scenario, where {\em the filled-in data are more 
reliable than the known data}. In the present manuscript we study
how far this idea can lead us.


\begin{figure}[tbhp]
\centering
\tabcolsep5pt
\begin{center}
\begin{tabular}{cccc}
\textbf{(a) original} & \textbf{(b) noisy} & 
\textbf{(c) mask} & \textbf{(d) inpainting}
\\[1.2mm]
\includegraphics[width=0.22\textwidth]{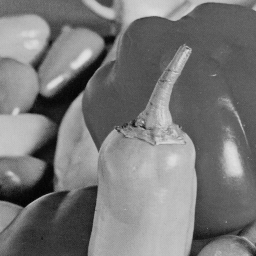} &
\includegraphics[width=0.22\textwidth]{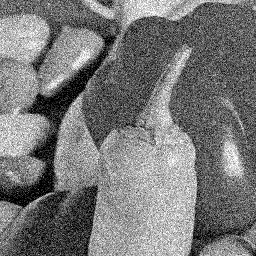} &
\includegraphics[width=0.22\textwidth]{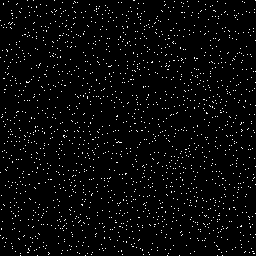} &
\includegraphics[width=0.22\textwidth]
{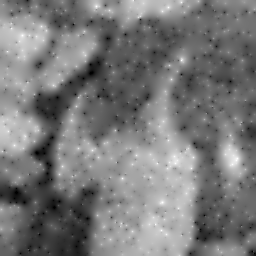} \\[-0.8mm]
\emph{peppers} & $\sigma_n=30$ & $3\%$ density & HD inpainting
\end{tabular}
\end{center}
\caption{\label{fig:introduction} Homogeneous diffusion (HD) inpainting on
the test image \emph{peppers} ($256 \times 256$ pixels, image range $[0, 255]$) 
with additive Gaussian noise of standard deviation $\sigma_n=30$ that we do not
clip. 
The mask pixels are randomly selected.
Note that the inpainted pixels are more reliable,
since they average noisy information from the neighborhood.
The visual difference is also reflected by the mean squared error (MSE):
The MSE of the noisy image in (b) is 904. 
Since the mask pixels are chosen randomly and are not changed by the 
inpainting, the MSE at mask pixel locations in (d) is still approximately 900.
However, the total image MSE in (d) is only 475.}
\end{figure}


\subsection{Our Contribution}
\label{subsec:contribution}

The goal of our work is to shed some light on the connections 
between PDE-based inpainting and denoising, two tasks which have coexisted 
for a long time, while their links have hardly been studied so far. 
We bridge this gap by a detailed investigation of the unconventional 
idea of denoising by inpainting.
To facilitate a rigorous mathematical analysis, we focus on 
homogeneous diffusion. As will be explained below, it constitutes the 
most transparent and most foundational setting in both worlds. 

The present paper builds upon our previous conference 
publication~\cite{APW17}, in which the basic denoising by inpainting 
framework is established. 
This framework reconstructs a denoised version of an image by 
averaging the results of multiple inpaintings obtained from distinct masks. 
Furthermore, two concrete implementations of this framework 
are proposed in~\cite{APW17}: 
The first uses shifted regular masks and allows to establish a relation 
between denoising by inpainting and classical diffusion filtering in 1D, 
while the second uses probabilistic densification to adapt the masks to 
the image structures and enables an edge-preserving denoising behavior.

We extend the aforementioned results by a much broader study of the 
framework in~\cite{APW17}, providing a fundamental understanding of
the connections between PDE-based image inpainting and denoising. 
Since denoising methods can also be used as plug-and-play priors in 
algorithms for solving inverse problems~\cite{LBA+22, REM17, VBW13},
our relations between inpainting and denoising approaches may have an 
even broader application spectrum.
Compared to~\cite{APW17}, we introduce the following additional 
contributions:
\smallskip
\begin{itemize}
\item We show that the heuristically motivated DbI framework 
      from~\cite{APW17} can be seen as a representative of a general 
      probabilistic framework, for which we derive a sound theory. 
      We argue that the denoising result obtained with such framework is an
      approximation of a minimum mean squared error (MMSE) estimate. 
\item We provide convergence estimates for the framework and propose a
     deterministic sampling approach to boost the convergence.
\item We prove a general relation between the mask density of regular 
      masks in the DbI framework and the diffusion time of homogeneous 
      diffusion filtering in 1D. We also propose an empirical generalization 
      of this result to 2D for uniform random masks.
\item We integrate a step that optimizes the gray values at the 
      selected mask pixels (tonal optimization) into the DbI framework. 
      We investigate its effect on the MMSE estimate and perform experiments 
      which confirm that tonal optimization can improve the denoising 
      performance of DbI in practice.
\item We show that the different spatial optimization approaches in the 
      DbI framework correspond to specific posterior distributions. 
      We compare two such strategies (the one presented in~\cite{APW17} and 
      a novel one) in terms of quality and provide the formulations for the 
      respective probability distributions. 
      Our experiments demonstrate that this data optimization leads to an 
      edge-preserving denoising behavior.
\item We replace homogeneous diffusion inpainting in the DbI framework 
      by biharmonic inpainting and show that it is unable to improve denoising 
      results. This confirms one of our key insights: 
      The hitherto hardly practiced data optimization can be as powerful as 
      widely used operator optimizations.
\end{itemize}


\paragraph{Why Homogeneous Diffusion\,?}
Our decision to focus on homogeneous diffusion is based on several
reasons:
\begin{itemize}
\item For denoising and image simplification, one should keep in mind 
      that homogeneous diffusion filtering is equivalent to Gaussian 
      convolution. The Gaussian is the only convolution kernel that is 
      separable and rotation invariant.
      The diffusion evolution generates a Gaussian scale-space
      representation \cite{Ii62, Li94, SNFJ96}, which is one of the most 
      widely-used scale-spaces and forms the basis of highly successful 
      interest point detectors such as SIFT \cite{Lo04a} and its 
      numerous variants.
\item In inpainting applications, homogeneous diffusion is particularly
      popular in inpainting-based compression \cite{GWWB08}, where one stores
      only a sparse subset of all pixels and reconstructs the image in
      the decoding phase by inpainting. By optimizing the stored data,
      homogeneous diffusion can achieve surprisingly faithful
      reconstructions \cite{MHWT12}. Moreover, its simplicity allows
      a detailed theoretical analysis \cite{BBBW08}, it frees the
      user from specifying parameters, and one can achieve real-time
      performance on current PC hardware even for large images \cite{KW22}.
\item Last but not least, there exist already well-understood connections 
      between diffusion processes for denoising and other approaches, 
      such as variational regularization methods \cite{NFD97,SchW98} 
      and wavelets \cite{WSMW05,DJS17}, but also deep neural network 
      architectures \cite{RH20,ASAP23}.
      Thus, establishing also connections to inpainting ideas gives more 
      comprehensive insights into various paradigms beyond diffusion-based
      denoising.
\end{itemize}
This discussion also implies that {\bf it is not the goal of the present 
paper to design novel approaches that outperform the most recent 
state-of-the-art approaches for denoising or inpainting.} This is 
reserved for future research that may benefit from the foundational 
insights in the our manuscript.


\subsection{Related Work}
\label{subsec:related}

Since we consider image inpainting as well as image denoising, we give an 
overview of some relevant methods from both fields and relate them to our 
work.


\smallskip
\paragraph{PDE-based Denoising and Inpainting}

We borrow several ideas from sparse PDE-based inpainting methods~\cite{GWWB08}. 
We mostly restrict ourselves to homogeneous diffusion inpainting~\cite{Ca88}, 
which can be implemented very 
efficiently~\cite{APMW21, CW21, HPW15, KN22, KW22, MBWF11}, and -- in spite 
of its simplicity -- can produce convincing results for suitably chosen 
data~\cite{BBBW09, BLPPR17, BHR21, CRP14, HSW13, Ho16, MHWT+11, OCBP14}.
Especially on piecewise constant images, such as cartoon images, 
depth maps or flow fields, homogeneous diffusion inpainting in 
conjunction with edge or segment information performs very 
well~\cite{Ca88, GLG12, HMWP13, JPW20, JPW21, LSJO12, MBWF11}.
This even allows some of these methods~\cite{JPW20, JPW21} to 
outperform HEVC~\cite{SOHW12} on such data.
Nonlinear diffusion inpainting methods, e.g., edge-enhancing diffusion (EED) 
inpainting~\cite{GWWB08, WW06}, can improve reconstruction quality 
for sparse inpainting, enabling lossy image 
codecs~\cite{GWWB08, JS21, SPMEWB14} competitive to JPEG~\cite{PM92} and 
JPEG2000~\cite{TM02}. 
On the other hand, such methods are more complex due to their nonlinearity. 
This complexity also carries over to the data optimization process.
Higher-order inpainting operators can also be used for sparse 
inpainting~\cite{BHS09, CRP14, GWWB08, SPMEWB14, THC11}, 
but can be more sensitive to noise. 
The quality of PDE-based sparse inpainting approaches 
strongly depends on the stored data, and in our denoising by inpainting 
framework we incorporate ideas from
spatial optimization~\cite{BBBW09, BLPPR17, CRP14, CW21, HSW13, 
Ho16, MBWF11, MHWT+11, OCBP14} and 
tonal optimization~\cite{CRP14, CW21, Ho16, MHWT+11, Pe19}.
To interpret the filtering results of the denoising by 
inpainting framework, we compare to classical diffusion-based image 
denoising methods. Aside from the simple homogeneous diffusion~\cite{Ii62}, 
we also consider methods that adapt the diffusion operator to the given 
image, namely linear space-variant diffusion~\cite{Fr92} and 
nonlinear diffusion~\cite{PM90}. 
We choose these methods because they are closest conceptually 
so we expect them to provide useful insights.


\smallskip
\paragraph{Patch-based Denoising and Inpainting}

Patch- or exemplar-based methods are another class of inpainting methods, 
and work especially well with textured data. The idea is to copy 
similar patches from known to unknown regions. Efros and Leung have proposed 
the first exemplar-based inpainting method~\cite{EL99b}, but many 
versions have been developed since then 
(e.g.,~\cite{AAD+17, ALM10, BSFG09, CPT04}), 
including the method of Facciolo et al.\ for sparse inpainting~\cite{FACS09}.
Inpainting approaches combining PDE-based and patch-based methods have also 
been presented~\cite{BVSO03, PW15}.
Inspired by the method of Efros and Leung~\cite{EL99b}, a patch-based 
denoising method called NL-means~\cite{BCM05a} has been proposed. It denoises 
an image based on a nonlocal weighted averaging of similar image patches. 
Other algorithms such as the famous BM3D algorithm~\cite{DFKE07} are also 
based on the filtering of image patches.
These observations further substantiate the ties between denoising and 
inpainting. The NL-means method can even be interpreted as a case of a 
denoising by inpainting approach, although it does not use the inpainting 
ideas as directly as we do. Of course, a direct application of patch-based 
inpainting techniques would lead to the copying of erroneous noisy data, 
and not to a denoising effect.


\smallskip
\paragraph{Sparse Signal Approximation}

A popular approach in the field of image denoising relies on the idea 
that signals (and images) can be represented as a linear combination of a 
smaller number of basis signals -- so-called atoms -- that are selected 
from a dictionary~\cite{El10}. Such a dictionary might for example consist
of the basis vectors of a suitable transform, that makes the signal 
representation sparse (e.g., a wavelet transform~\cite{Ma98} or a discrete 
cosine transform (DCT)~\cite{ANR74}).
The task is to then find those atoms, that best represent the given 
signal~\cite{CDS98, DJ94, MZ93}.
To fill in missing information in images, several authors also consider
sparse representations in some transform domain such as the DCT~\cite{Gu02} 
or the shearlet domain~\cite{KKL13}. This shows another bridge between
the two tasks of denoising and inpainting.
Hoffmann et al.~\cite{HPW15} relate linear PDE-based inpainting methods 
to concepts from sparse signal approximation.
They solve the inpainting problem with the help of discrete Green's 
functions~\cite{BL58, CY00}, which can be interpreted as atoms in a 
dictionary. 
This allows for a sparse representation of the inpainting solution. 
Kalmoun et al.~\cite{KN22} follow a similar approach by solving 
homogeneous diffusion inpainting with the charge simulation 
method~\cite{Ka99, KA64}.
An application of homogeneous diffusion inpainting with Green's functions 
is the video codec by Andris et al.~\cite{APMW21}. 
We justify certain design choices within the DbI framework with 
results from this field. Notably, homogeneous diffusion inpainting 
is based on the idea that the Laplacian of the reconstructed image 
is mostly sparse. On the other hand, the DbI framework combines 
multiple noisy sparse representations in order to get a denoised 
but non-sparse representation. The latter can be studied rigorously 
from a Bayesian denoising perspective, which is why we discuss 
this next.


\smallskip
\paragraph{Bayesian Denoising}

The study of denoising has also been carried out from a probabilistic 
perspective. Here, the assumption is that some prior information regarding 
the noise distribution and/or the image distribution is available. 
This can be incorporated in a denoising framework through Bayes' 
rule, such that the final denoised result is conditioned on this 
information about the distributions. The latter provides a 
correspondence between classical denoising 
variational methods and specific Bayesian priors~\cite{EKV23,HKU02,LCB+12}.
The standard approach is to employ statistical inference approaches, 
such as maximum likelihood (ML) estimation, maximum a posteriori (MAP) 
estimation, or minimum mean squared error (MMSE) estimation. Both 
the MAP and MMSE approach rely on a posteriori density, and as such 
they require a model of the distribution of considered classes of 
images.
One of the first such models uses a Gibbs distribution for the 
prior~\cite{GG84}. Subsequently, a number of 
works have built upon this idea. The most relevant to our setting 
is that by Larsson and Selen~\cite{LS07}, which studies MMSE 
estimation in the context of sparse vector representations. 
Our sparse inpaintings can be interpreted as such sparse vector 
representations. 
Moreover, in the current work we show that the averaging performed 
in~\cite{APW17} is in fact a Monte Carlo approach to approximate an 
MMSE estimate. 


\smallskip
\paragraph{Cross-Validation}

We also see the work of Craven and Wahba~\cite{CW78} on (generalized) 
cross-validation as conceptually related to parts of our work.
Cross-validation can be used to optimize parameters in denoising
models~\cite{CW78, LS07, WW98}. It removes data points from given noisy 
observations and judges the quality of a parameter selection in terms 
of the model's capability to reconstruct the data at these locations.
Related ideas are also pursued in \cite{CHN05}.
Probabilistic densification~\cite{HMWP13} and sparsification~\cite{MHWT+11},
two concepts from spatial optimization that we consider in our framework,
also use the error of the inpainted reconstruction at left out locations --
in our case also on noisy data.
Yet, both applications differ, as the goal of the latter methods is to
construct an inpainting mask and not to optimize model parameters.


\smallskip
\paragraph{Neural Denoising and Inpainting}

In recent years, many very powerful methods for inpainting and denoising
have been proposed that rely on neural networks. They are, however, not
a topic of our paper, since we aim at gaining structural insights into 
the connections between inpainting and denoising. 
Such results on classical approaches are still relevant in the learning 
era~\cite{EKV23}. 
They may serve as foundations for deep learning-based methods, and   
model- and learning-based approaches may be fused to obtain powerful and 
transparent algorithms.
It is our hope that in the long run, our insights can also be 
beneficial to neural approaches.


\subsection{Paper Organization}
\label{subsec:organization}

In \Cref{sec:diffusion} we briefly 
introduce the basic idea behind diffusion filtering and its application to 
image denoising and image inpainting. In \Cref{sec:framework} we present the 
framework for denoising by inpainting from~\cite{APW17} and show that it can 
be interpreted as a Monte Carlo approach for approximating an MMSE estimate. 
We additionally provide convergence results, and suggest a method to 
boost the convergence by employing low-discrepancy sequences.
In \Cref{sec:rel-dbi-hd} we relate denoising by inpainting with non-adaptive
masks to classical diffusion filtering. 
In \Cref{sec:adaptive} we present strategies for adaptively selecting the 
mask pixels in the DbI framework, which leads to space-variant denoising 
behavior.
Our experiments and results are presented in \Cref{sec:experiments}, 
and we conclude the paper in \Cref{sec:conclusion}.


\section{Basics of Diffusion Filtering}
\label{sec:diffusion}

In its original context of physics, diffusion is a process that equilibrates 
particle concentrations. When working with images, we interpret the
gray values as particle concentrations and use diffusion processes as 
smoothing filters that balance gray value differences.
To this end, we define the original grayscale image as a function 
$f:\Omega \to \R$, with $\Omega \subset \R^2$ being a rectangular 
image domain. Similarly, $u:\Omega \times [0,\infty) \to \R$ denotes 
the evolving, filtered image. 
Then the diffusion evolution is described by the following PDE:
\begin{equation}
\label{eq:diff-pde}
\partial_t u(\bm{x}, t) = \opdiv(g \bm{\nabla} u(\bm{x}, t)) 
    \quad\textnormal{for }\bm{x} \in \Omega, \,\,t \in (0, \infty).
\end{equation}
Here $t$ denotes time, $\bm \nabla=(\partial_x, \partial_y)^\T$
is the spatial gradient, $\opdiv(\bm{v}) = \partial_x v_x + \partial_y v_y$ 
is the spatial divergence, and the scalar diffusivity $g$ determines the 
local smoothing activity. We discuss different choices of $g$ in
\Cref{subsec:denoising}.
Note that $g$ can be extended to a diffusion tensor to introduce anisotropy 
into the process~\cite{We97}, but since we do not consider such a case in 
this paper, we refrain from discussing it here.
We equip the PDE with an initial condition at time $t=0$ and reflecting 
boundary conditions at the image boundary $\partial \Omega$:
\begin{alignat}{3}
\label{eq:diff-ic} u(\bm{x}, 0) &= f(\bm{x}) 
    \quad &\textnormal{for }& \bm{x} \in \Omega, \\
\label{eq:diff-bc} \partial_{\bm{n}} u(\bm{x}, t) &= 0 
    \quad &\textnormal{for }& 
        \bm{x} \in \partial \Omega, \,\,t \in (0, \infty),
\end{alignat}
where $\bm{n}$ is the outer normal vector at the image boundary.
Solving this initial boundary value problem for $u$ yields a family of 
filtered images $\{u(\cdot, t) \mid t \geq 0\}$.


\subsection{Diffusion for Image Denoising}
\label{subsec:denoising}

In image denoising the image $f$ is a noisy version of the noise-free 
ground truth image $f_r$. 
In our case we assume zero-mean additive white Gaussian noise, i.e., 
$f = f_r + n$ with $n\in \mathcal{N}(0, \sigma_n^2)$. 
Diffusion processes are good candidates for image denoising tasks thanks to 
their smoothing properties. Depending on the form of the diffusivity $g$, 
different processes are obtained.


\subsubsection{Homogeneous Diffusion}
\label{subsub:homogeneous}

By setting $g \equiv 1$, \labelcref{eq:diff-pde} simplifies to 
$\partial_t u = \Delta u$, with 
$\Delta u = \partial_{xx} u + \partial_{yy} u$ being the Laplacian operator. 
The resulting process is known as \emph{homogeneous diffusion}~\cite{Ii62}. 
Its analytical solution in the unbounded image domain $\R^2$ is given by a 
convolution of the original image with a Gaussian kernel 
$K_\sigma$ with standard deviation $\sigma = \sqrt{2t}$. 
The resulting images $\{u(\cdot, t) \mid t \geq 0\}$ constitute the 
so-called Gaussian scale-space~\cite{Ii62, WII97}.
Since $g$ is selected to be constant, the smoothing strength is the same 
across the entire image. Therefore, not only the noise is reduced, but also 
semantically important image structures such as edges are smoothed.


\subsubsection{Linear Space-Variant Diffusion}
\label{subsub:linear-space-var}

To overcome the drawbacks of homogeneous diffusion, one can make the process 
space-variant by selecting a diffusivity function that varies depending on the 
structure of the \emph{initial} image $f$~\cite{Fr92}. 
This is called \emph{linear space-variant diffusion}.
If edges and other high-gradient features are to be preserved, the diffusivity 
should be decreasing with increasing gradient magnitude of the image, 
so that that the smoothing would be reduced at edges.
An example for a suitable function is the Charbonnier 
diffusivity~\cite{CBAB97}:
\begin{equation}
g(|\bm{\nabla} f|^2) 
= \frac{1}{\sqrt{1+\frac{|\bm{\nabla} f|^2}{\lambda^2}}},
\end{equation}
where $| \cdot |$ denotes the Euclidean norm. 
The contrast parameter $\lambda > 0$ is used to distinguish locations 
where smoothing should be applied (for $|\bm{\nabla} f| \ll \lambda$, 
we get $g_{\lambda} \to 1$) and locations where it should be reduced 
(for $|\bm{\nabla} f| \gg \lambda$, we obtain $g_{\lambda} \to 0$).


\subsubsection{Nonlinear Diffusion}
\label{subsub:nonlinear}

Alternatively, one can make the diffusivity function $g$ dependent on the 
\emph{evolving} image $u$. This allows to update the locations 
where smoothing is reduced during the evolution, by choosing them based 
on the image $u$, which becomes gradually smoother and less noisy.
The resulting process 
$\partial_t u = \opdiv(g(| \bm{\nabla} u |^2) \bm{\nabla} u)$ 
is \emph{nonlinear}~\cite{PM90}. 
The feedback mechanism throughout the evolution helps steering the process 
to achieve better results.


\subsection{Diffusion for Image Inpainting}
\label{subsec:diff-inp}

Diffusion processes can also be used to fill in missing information in 
images~\cite{Ca88, CS01a, WW06}. Particularly, they allow to reconstruct an 
image from only a small number of pixels by propagating information from known 
to unknown areas~\cite{GWWB08}. The set of known pixels is called the 
\emph{inpainting mask} and is denoted by $K \subset \Omega$.
To recover the image, the information at the unknown locations is computed 
as the steady state ($t \to \infty$) of a diffusion process, while the values 
at mask locations are preserved.
The parabolic inpainting formulation is obtained by modifying 
\labelcref{eq:diff-pde,eq:diff-ic} accordingly:
\begin{alignat}{3}
\label{eq:inp-pde-par}
\partial_t u(\bm{x}, t) &= \opdiv(g \bm{\nabla} u(\bm{x}, t)) 
    \quad &\textnormal{for }& 
        \bm{x} \in \Omega \setminus K, \,\,t \in (0, \infty), \\
u(\bm{x}, t) &= f(\bm{x}) 
    \quad  &\textnormal{for }& \bm{x} \in K\, ,\,t \in [0, \infty), \\
u(\bm{x}, 0) &= 0
    \quad  &\textnormal{for }& \bm{x} \in \Omega\setminus K, \\
\partial_{\bm{n}} u(\bm{x}, t) &= 0 
    \quad &\textnormal{for }& 
        \bm{x} \in \partial \Omega, \,\,t \in (0, \infty).
\end{alignat}
For $g\equiv 1$, \labelcref{eq:inp-pde-par} is the homogeneous diffusion 
PDE~\cite{Ii62} and we talk about \emph{homogeneous diffusion inpainting}
(also called \emph{harmonic inpainting}). 
We almost exclusively consider homogeneous diffusion inpainting in the 
remainder of this paper, so we set $g\equiv 1$ in the following.
Instead of computing the steady state of the parabolic diffusion equation, 
we may solve the corresponding boundary value problem:
\begin{alignat}{3}
\label{eq:inp-pde-ell} -\Delta u(\bm{x}) &= 0 
    \quad &\textnormal{for }& \bm{x} \in \Omega \setminus K, \\
\label{eq:inp-init-ell}  u(\bm{x}) &= f(\bm{x}) 
    \quad &\textnormal{for }& \bm{x} \in K, \\
\partial_{\bm{n}} u(\bm{x}) &= 0 
    \quad &\textnormal{for }& \bm{x} \in \partial \Omega.
\end{alignat}
The problem may be written equivalently using the variational 
formulation
\begin{equation}
\label{eq:harmonic_variational_formulation}
\min_{u}\int_{\Omega}|\nabla u(\bm{x})|^2\,d\bm{x}, 
\text{ such that } u(\bm{x}) = f(\bm{x}) \text{ for } \bm{x} \in K.
\end{equation}
This suggests the interpretation that the inpainting is designed to 
penalize the gradient magnitude of the reconstruction, i.e., it inherently 
promotes smoothness.
In order to simplify the discretization of the boundary value problem 
formulation, we introduce a mask indicator function $c = \mathds{1}_{K}$ 
(we use the term \emph{mask} synonymously for the set $K$ and the function 
$c$), that takes the value 1 at points from $K$ and 0 elsewhere. 
This allows us to combine \labelcref{eq:inp-pde-ell,eq:inp-init-ell} into a 
single equation
\begin{equation}
\label{eq:inp-pde-single}
\bigl(c(\bm{x}) + (1 - c(\bm{x})) (-\Delta)\bigr)u(\bm{x}) = c(\bm{x})f(\bm{x}) 
    \quad \text{for } \bm{x} \in \Omega.
\end{equation}


\subsection{Discrete Homogeneous Diffusion Inpainting}
\label{subsec:discrete-inp}

Since we are working with digital images, the above considerations need to be 
translated to the discrete setting. We therefore discretize the images on a 
regular pixel grid of size $n_x \times n_y$. 
Then we write them as vectors of length $N = n_x n_y$ that are obtained by 
stacking the discrete images column-by-column, e.g., $\bm{f}, \bm{u} \in \R^N$.
Furthermore, let $\bm{L}\in\R^{N\times N}$ denote the five-point stencil 
discretization matrix of the negated Laplacian $(-\Delta)$ with reflecting 
boundary conditions $\partial_{\bm{n}}u(\bm{x}) = 0$ for 
$\bm{x}\in\partial\Omega$. Additionally, let $\bm{C}=\diag(\bm{c})$ be the 
diagonal matrix with the mask vector $\bm{c}\in\{0,1\}^N$ discretizing $c$, 
and let $\bm{I}$ be the $N \times N$ identity matrix. 
Then the discrete version of \labelcref{eq:inp-pde-single} can be formulated 
as the linear system of equations:
\begin{equation}
    \left(\bm{C}+(\bm{I}-\bm{C})\bm{L}\right)\bm{u} = \bm{C}\bm{f},
\end{equation}
and the reconstruction can be written explicitly as
\begin{equation}
\label{eq:discrete_harmonic_bvp_solution}
\bm{u} = \bm{r}(\bm{c}, \bm{f}) = 
    \left(\bm{C}+(\bm{I}-\bm{C})\bm{L}\right)^{-1}\bm{C}\bm{f}. 
\end{equation}
The inverse of the \emph{inpainting matrix} 
$\bm{M}_{\bm{c}} \coloneqq \bm{C}+(\bm{I}-\bm{C})\bm{L}$ exists as long as 
$\bm{C}\ne \bm{0}$~\cite{MBWF11}. To deal with the case 
$\bm{C}=\bm{0}$ 
we define $\bm{r}(\bm{0},\bm{f}) \coloneqq 
\frac{1}{N} \bm{1}^{\T}\bm{f}$, i.e., we take the average.
If we want to approximate 
the image $\bm{f}$ instead of interpolating it over $\bm{C}$, we can 
replace $\bm{C}\bm{f}$ with $\bm{C}\bm{g}$, where
\begin{equation}
\label{eq:tonal-prob}
    \bm{g} \in \argmin_{\bm{h}\,:\,\bm{h}|_{\bar{\bm{c}}} = \bm{0}} 
    \|\bm{r}(\bm{c},\bm{h})-\bm{f}\|^2_2.
\end{equation}
Here $\bm{h}|_{\bm{c}}$ is the restriction of $\bm{h}$ to $\bm{c}$ and 
$\bm{h}|_{\bar{\bm{c}}}$ is the restriction of $\bm{h}$ to the 
complement $\bar{\bm{c}} = \bm{1}-\bm{c}$. The optimization is thus only 
over $\bm{h}|_{\bm{c}}$ since the remainder of the values are irrelevant 
for the inpainting result, so we set them to zero.
The least squares problem is known as the \emph{tonal optimization} problem 
and we discuss its implications for the current work in 
\Cref{ss:MMSE_and_Tonal_Optimization}.
Additionally, we observe that the reconstruction is linear in $\bm{g}$. 
This motivates us to write it as a linear combination of 
basis vectors with weights given by $\bm{g}|_{\bm{c}}$. 
Let $\bm{B}_{\bm{c}} \coloneqq (\bm{M}_{\bm{c}}^{-1})|_{\bm{I}\times \bm{C}}$ 
be the restriction of $\bm{M}^{-1}_{\bm{c}}$ to the columns corresponding to 
non-zeros in $\bm{c}$, and we set $m=\|\bm{c}\|_0$ to be the number of non-zeros 
in $\bm{c}$. By denoting the columns as 
$\left\{\bm{b}_{\bm{c}}^k\right\}_{k=1}^{m}$, i.e., 
$\bm{B}_{\bm{c}} = \left[\bm{b}_{\bm{c}}^1 \, \ldots \, 
\bm{b}_{\bm{c}}^m \right]$, 
we can write the reconstruction as
\begin{equation}
\label{eq:discrete_harmonic_bvp_solution_basis}
\bm{u} = \bm{r}(\bm{c}, \bm{g}) = \bm{M}^{-1}_{\bm{c}}\bm{C}\bm{g} 
= \bm{B}_{\bm{c}}\,\bm{g}|_{\bm{c}} = 
    \sum_{k=1}^{m} (\bm{g}|_{\bm{c}})_k \,\bm{b}^k_{\bm{c}}.
\end{equation}
We see that the columns of $\bm{B}_{\bm{c}}$ are the basis vectors 
induced from $\bm{r}$ and $\bm{c}$. They are also termed 
\emph{inpainting echoes}~\cite{DN01,MHWT+11}. We note that 
inpainting with $\bm{g}|_{\bm{c}} = \bm{f}|_{\bm{c}}$ constructs 
the interpolant over $\bm{c}$ in the space 
$\spanV (\bm{B}_{\bm{c}}) \subseteq \mathbb{R}^N$.
Since the tonal optimization solution can be written as 
$\bm{g}|_{\bm{c}} = (\bm{B}_{\bm{c}})^{+}\bm{f}$, where 
$(\bm{B}_{\bm{c}})^{+}$ is the Moore-Penrose pseudo-inverse, we note that 
$\bm{r}(\bm{c},\bm{g}) = \bm{B}_{\bm{c}}(\bm{B}_{\bm{c}})^{+}\bm{f}$ is 
the orthogonal projection of $\bm{f}$ on the subspace $\spanV (\bm{B}_{\bm{c}}) 
\subseteq \mathbb{R}^N$, i.e., the best approximant of $\bm{f}$ in this 
space.


\section{Our Denoising by Inpainting Framework}
\label{sec:framework}

We now present the basic idea and the framework for denoising by 
inpainting proposed in our conference paper~\cite{APW17}.
Since the framework inherently links inpainting and denoising, it is 
well-suited to study connections between the two tasks.
As previously mentioned, we use diffusion-based inpainting -- 
specifically homogeneous diffusion inpainting -- for image denoising, 
by only keeping a sparse subset of the noisy input data and by 
reconstructing the rest. Inpainting on noisy images differs from the 
classical setting and poses additional challenges. During the inpainting 
process, gray values at mask locations are not altered. As they might 
contain errors from the noise, these mask pixels are less trustworthy 
than inpainted pixels, which combine information from their surrounding 
mask pixels. While we want to exploit the filling-in effect in 
unknown areas, this observation implies that a single inpainted image 
cannot give satisfactory denoising results.
Therefore, we compute multiple inpaintings with different masks and obtain
the final result by averaging them. This ensures that none of the pixels
remain unchanged (unless a pixel is contained in all masks).
In the current work, we further mitigate the issue of noisy mask pixels by
employing tonal optimization (see \Cref{ss:MMSE_and_Tonal_Optimization}).
If we denote the $n$ different masks by $\{\bm{c}^\ell\}_{\ell=1}^{n}$,
we can generate the inpaintings $\{\bm{v}^\ell\}_{\ell=1}^{n}$ via 
\begin{equation}
\bm{v}^\ell = \bm{r}(\bm{c}^\ell, \bm{f}) = 
\left(\bm{C}^\ell+\left(\bm{I}-\bm{C}^\ell\right)\bm{L}\right)^{-1}
\bm{C}^\ell\bm{f}. 
\end{equation}
We obtain the final denoising result $\langle\bm{u}\rangle_n$ by averaging:
\begin{equation}
\label{eq:avg}
\langle \bm{u} \rangle_{n}= \frac{1}{n} \sum_{\ell=1}^{n} \bm{v}^\ell 
= \frac{1}{n}\sum_{\ell=1}^n\bm{r}(\bm{c}^\ell,\bm{f}).
\end{equation}
As we fix the inpainting operator (for a discussion of denoising by 
biharmonic inpainting see \Cref{subsec:exp-biharmonic}), the only freedom 
in the framework lies in the selection of the different masks. 
This is in contrast to the common strategy in denoising, where
all available data is used and the operator is optimized instead.
To study the effects of different data selection strategies, we will 
borrow several ideas from mask optimization for image compression. 
To obtain multiple different masks as our framework requires, 
we rely on some degree of randomness in the mask generation processes 
(see \Cref{sec:adaptive}). Since we make use of stochastic strategies, 
we formalize and study DbI from a probabilistic point of view 
in the following subsection.


\subsection{Probabilistic Theory}
\label{subsec:prob-theory}

As seen in \labelcref{eq:avg}, the denoised image is the result of averaging 
$n$ inpaintings from $n$ different masks, that are generated by some mask 
optimization process. In the following, we interpret this from a 
probabilistic point of view.
This allows us to formalize the DbI framework from our 
conference paper~\cite{APW17} and provides us with tools to study 
and boost the convergence of our methods in \Cref{subsubsec:convergence} and
\Cref{subsubsec:low-discrepancy}, respectively.
We take the masks $\{\bm{c}^\ell\}_{\ell=1}^{n}$ to be independent and 
identically distributed samples from a predetermined distribution  
conditioned on $\bm{f}$, with a conditional probability mass function (PMF) 
$p(\bm{c}|\bm{f})$.
Then the estimator $\bm{u}$ converges to the following conditional 
expectation for $n\to\infty$:
\begin{equation}
\label{eq:expectation_denoising}
\mathbb{E}[\langle\bm{u}\rangle_n|\bm{f}] = 
\mathbb{E}\left[\frac{1}{n}
\sum_{\ell=1}^n\bm{r}(\bm{c}^\ell,\bm{f})\biggr|\bm{f}\right] 
= 
\frac{1}{n}\sum_{\ell=1}^n\mathbb{E}[\bm{r}(\bm{c},\bm{f})|\bm{f}] 
 = 
\sum_{\bm{c}\in\{0,1\}^N}\bm{r}(\bm{c}, \bm{f})\,p(\bm{c}|\bm{f}).
\end{equation}
The second equality holds because the masks were assumed to be identically 
distributed, and thus $\mathbb{E}[\bm{r}(\bm{c}^{\ell},\bm{f})|\bm{f}] = 
\mathbb{E}[\bm{r}(\bm{c},\bm{f})|\bm{f}]$ for any $\bm{c}$ sampled with 
the same PMF $p$. The fourth equality follows from the definition of 
the conditional mathematical expectation.
We note that from this probabilistic point of view, spatial adaptivity is
provided through the design of the PMF $p$.
The following proposition shows that the DbI result constitutes a minimum 
mean squared error (MMSE) estimate. This emphasizes its optimality under
certain assumptions.


\smallskip
\begin{proposition}[DbI as an MMSE Estimate]
\label{thm:dbi-mmse-int} 
The expectation \labelcref{eq:expectation_denoising} of the DbI averaging 
\labelcref{eq:avg} can be interpreted as an MMSE estimate under prior assumptions 
on the image and noise distributions, i.e., it solves the minimization problem
\begin{equation}
\label{eq:MMSE_interpolation_min_formulation}
    \min_{\bm{u}\in\mathbb{R}^N} 
    \mathbb{E}[\|\bm{u}-\bm{w}\|^2_2|\bm{f}]
    =
    \min_{\bm{u}\in\mathbb{R}^N} 
    \mathbb{E}[\|\bm{u}-\bm{r}(\bm{c},\bm{f})\|^2_2|\bm{f}].
\end{equation}
\end{proposition}


\begin{proof}
We can rewrite the minimization problem 
\labelcref{eq:MMSE_interpolation_min_formulation} as
\begin{equation}
    \min_{\bm{u}\in\mathbb{R}^N} 
    \mathbb{E}[\|\bm{u}-\bm{r}(\bm{c},\bm{f})\|^2_2|\bm{f}] 
    = \min_{\bm{u}\in\mathbb{R}^N} 
    \sum_{\bm{c}\in\{0,1\}^N}
    \|\bm{u}-\bm{r}(\bm{c},\bm{f})\|^2_2\,p(\bm{c}|\bm{f}).
\end{equation}
Taking the derivative w.r.t.\ $\bm{u}$ and setting it to zero results in 
the MMSE estimate
\begin{equation}
    \label{eq:MMSE_estimate_interpolation}
    \bm{u}^{\text{MMSE}} = \mathbb{E}[\bm{r}(\bm{c},\bm{f})|\bm{f}] = 
    \sum_{\bm{c}\in\{0,1\}^N}\bm{r}(\bm{c}, \bm{f})\,p(\bm{c}|\bm{f}).
\end{equation}
By \labelcref{eq:expectation_denoising} this is the same as the expectation
$\mathbb{E}[\langle \bm{u}\rangle_n]$ of the DbI estimator 
$\langle \bm{u}\rangle_n$.
\end{proof}


The estimate $\bm{u}^{\text{MMSE}}$ is close to $\bm{f}_r$ (and 
$\langle \bm{u}\rangle_n$ is close 
to $\bm{f}_r$), whenever $\bm{v}=\bm{r}(\bm{c},\bm{f})$ with 
$\bm{c}\sim p(\bm{c}|\bm{f})$ provides a good model for  
the distribution from which $\bm{f}_r$ is assumed to originate. 
This formalization of DbI as an estimator for the MMSE estimate 
therefore provides an additional justification for the DbI framework 
as an image denoising approach.


\subsubsection{MMSE and Tonal Optimization}
\label{ss:MMSE_and_Tonal_Optimization}

The classical DbI formulation \labelcref{eq:avg} from~\cite{APW17} employs 
an \emph{interpolating} inpainting. 
It is natural to extend the framework
to the best \emph{approximating} inpainting, computing the denoised image 
$\langle\bm{u}\rangle_n$ as
\begin{equation}
\label{eq:avg-appr}
\langle\bm{u}\rangle_n = 
\frac{1}{n}\sum_{\ell=1}^n\bm{r}(\bm{c}^\ell,\bm{g^\ell}),
\end{equation}
where the masks $\{\bm{c}^\ell\}_{\ell=1}^{n}$ are selected as before, while
$\{\bm{g}^\ell\}_{\ell=1}^{n}$ are the 
solutions to the corresponding tonal optimization problems 
\labelcref{eq:tonal-prob}.
Next we show that after relaxing assumptions on the gray 
values compared to \Cref{thm:dbi-mmse-int}, the MMSE estimate actually 
corresponds to DbI with an approximating inpainting instead of an 
interpolating one.


\smallskip
\begin{proposition}[DbI with Approximating Inpainting as an MMSE Estimate]
\label{thm:dbi-mmse-appr} 
The DbI result based on a best approximating inpainting 
\labelcref{eq:avg-appr} can also be interpreted as an 
MMSE estimate, assuming that the gray values $\bm{h}$ are now also 
a random variable conditioned on $\bm{f}$.
\end{proposition}


\begin{proof}
Firstly, we note that the minimization problem for the MMSE now differs, as
the expectation has to be taken over the gray values $\bm{h}$ as well:
\begin{equation}
\begin{split}
   \min_{\bm{u}\in\mathbb{R}^N}
    \mathbb{E}[\|\bm{u}-\bm{w}\|^2_2|\bm{f}] 
    =&\min_{\bm{u}\in\mathbb{R}^N}
    \mathbb{E}[\|\bm{u}-\bm{r}(\bm{c},\bm{h})\|^2_2|\bm{f}] \\
    = &\min_{\bm{u}\in\mathbb{R}^N} 
    \sum_{\bm{c}\in\{0,1\}^N} 
    \mathbb{E}[\|\bm{u}-\bm{r}(\bm{c},\bm{h})\|^2_2|\bm{f},\bm{c}]
    \, p(\bm{c}|\bm{f}) \\
    = &\min_{\bm{u}\in\mathbb{R}^N}
    \sum_{\bm{c}\in\{0,1\}^N} \left(\int_{\bm{h}\in\mathbb{R}^{N}}
    \|\bm{u}-\bm{r}(\bm{c},\bm{h})\|^2_2 
    \,p(\bm{h}|\bm{f},\bm{c})\,d\bm{h}\right)
    p(\bm{c}|\bm{f}).
\end{split}
\end{equation}
As before, differentiation w.r.t.\ $\bm{u}$ yields the MMSE estimate
\begin{equation}
    \bm{u}^{\text{MMSE}} = \mathbb{E}[\bm{r}(\bm{c},\bm{h})|\bm{f}] 
    = 
    \sum_{\bm{c}\in\{0,1\}^N} 
    \mathbb{E}[\bm{r}(\bm{c},\bm{h})|\bm{f},\bm{c}]
    \, p(\bm{c}|\bm{f}),
\end{equation}
which is similar to \labelcref{eq:MMSE_estimate_interpolation},
but now contains the expectation
\begin{equation}
    \mathbb{E}[\bm{r}(\bm{c},\bm{h})|\bm{f},\bm{c}] =
    \int_{\bm{h}\in\mathbb{R}^{N}}
    \bm{r}(\bm{c},\bm{h})\,p(\bm{h}|\bm{f},\bm{c})\,d\bm{h}.
\end{equation}
To compute $\mathbb{E}[\bm{r}(\bm{c},\bm{h})|\bm{f},\bm{c}]$, we need 
to know the a posteriori density $p(\bm{h}|\bm{f},\bm{c})$. 
If we assume that the noise is normally distributed $\bm{n} = 
(\bm{r}(\bm{c},\bm{h})-\bm{f}) \sim 
\mathcal{N}(\bm{0},\sigma^2_{\bm{n}}\bm{I})$,
and that the gray values restricted to the mask $\bm{h}|_{\bm{c}}$ 
are normally distributed 
$\bm{h}|_{\bm{c}}\sim\mathcal{N}(\bm{0}, \sigma^2_{\bm{h}|_{\bm{c}}}\bm{I})$, 
then the expectation can be calculated~\cite{LS07} as
\begin{equation}
    \label{eq:estimate_tonal_opt}
    \mathbb{E}[\bm{r}(\bm{c},\bm{h})|\bm{f},\bm{c}] = 
    \bm{B}_{\bm{c}}\,\mathbb{E}[\bm{h}|_{\bm{c}}|\bm{f},\bm{c}]
    = \bm{B}_{\bm{c}}
    \left(\frac{\sigma^2_{\bm{n}}}{\sigma_{\bm{h}|_{\bm{c}}}^2}\bm{I}
    + \bm{B}_{\bm{c}}^{\T}\bm{B}_{\bm{c}}\right)^{-1}\bm{B}_{\bm{c}}^{\T} 
    \bm{f}.
\end{equation}
Since we do not know $\sigma_{\bm{h}|_{\bm{c}}}$ and because the assumption of 
the normality of the gray values may not be a very plausible one, we 
can dispense away with it by taking $\sigma_{\bm{h}|_{\bm{c}}}\to\infty$, 
which results in a tonally optimized inpainting:
\begin{equation}
    \lim_{\sigma_{\bm{h}|_{\bm{c}}}\to\infty}
    \mathbb{E}[\bm{r}(\bm{c},\bm{h})|\bm{f},\bm{c}] = 
    \bm{B}_{\bm{c}}\lim_{\sigma_{\bm{h}|_{\bm{c}}}\to\infty}
    \left(\frac{\sigma^2_{\bm{n}}}{\sigma_{\bm{h}|_{\bm{c}}}^2}\bm{I}
    + \bm{B}_{\bm{c}}^{\T}\bm{B}_{\bm{c}}\right)^{-1}\bm{B}_{\bm{c}}^{\T} 
    \bm{f}
    =
    \bm{B}_{\bm{c}} (\bm{B}_{\bm{c}})^{+}\bm{f}.
\end{equation}
Using $\bm{B}_{\bm{c}}(\bm{B}_{\bm{c}})^+\bm{f} 
= \bm{r}(\bm{c}, (\bm{B}_{\bm{c}})^+\bm{f})$,
the new MMSE estimate differs with \labelcref{eq:MMSE_estimate_interpolation} 
only in that we have approximation instead of interpolation:
\begin{equation}
    \label{eq:MMSE_estimate_approximation}
    \bm{u}^{\text{MMSE}}
    = 
    \sum_{\bm{c}\in\{0,1\}^N} \mathbb{E}[\bm{r}(\bm{c},\bm{h})|\bm{f},\bm{c}]
    \, p(\bm{c}|\bm{f}) 
    =
    \sum_{\bm{c}\in\{0,1\}^N} \bm{r}(\bm{c},(\bm{B}_{\bm{c}})^{+}\bm{f})
    \, p(\bm{c}|\bm{f}).
\end{equation}
This corresponds exactly to the expectation of the approximating DbI 
formulation.
\end{proof}


We note that the above analysis did not require $\bm{r}$ to be linear in 
$\bm{f}$ except for the approximation of $\bm{f}$. Given a fixed $\bm{c}$, 
a natural extension to nonlinear operators could use nonlinear 
least-squares to compute something similar to $\bm{B}_{\bm{c}}^+\bm{f}$.
By using the approximating formulation, we project the image onto the various 
sub-spaces induced by the inpainting operator $\bm{r}$ and the mask $\bm{c}$. 
We will show in \Cref{subsubsec:exp-tonal-optimization} that in practice, 
tonal optimization is able to improve quality and to reduce the variance of 
MMSE denoising, since it mitigates the error from the interpolation of 
noisy mask pixels and provides representations that are closer to $\bm{f}$ 
in terms of MSE.


\subsubsection{Interpreting Tonal Optimization as MAP Estimate} 

Not directly related to the classical averaging formulation of DbI, but 
nevertheless interesting and a valuable extension, is the fact that 
spatial and tonal optimization for a single inpainting can also be 
framed as a maximum a posteriori (MAP) estimate. 
In MAP estimation, instead of minimizing the MSE, we want to 
find an inpainting $\bm{w}$ that maximizes the posterior:
\begin{equation}
    \label{eq:MAP_argmax}
    \argmax_{\bm{w}} p(\bm{w}|\bm{f}) 
    = \argmax_{\bm{c},\bm{h}} p(\bm{h},\bm{c}|\bm{f}) 
    = \argmax_{\bm{c},\bm{h}} p(\bm{f}|\bm{h},\bm{c}) 
    p(\bm{h}|\bm{c}) p(\bm{c}).
\end{equation}
We have assumed that $\bm{w}=\bm{r}(\bm{c},\bm{h})$ is an injection,
so we have $p(\bm{w}|\bm{f}) = p(\bm{r}(\bm{c},\bm{h})|\bm{f}) =
p(\bm{h},\bm{c}|\bm{f})$. In the non-injective case one gets a set 
\begin{equation}
    p(\bm{w}|\bm{f}) = p(\bm{r}^{-1}(\bm{w})|\bm{f}) = 
    p(\{\bm{h},\bm{c}\,:\,\bm{w}=\bm{r}(\bm{c},\bm{h})\}|\bm{f}),
\end{equation}
which does not change the derivation meaningfully, except for introducing 
additional technical details. Thus, for the sake of clarity, we proceed with 
the injective case, but a similar argument holds in the general setting.
The maximization problem (\ref{eq:MAP_argmax}) can be split into two 
optimization problems:
\begin{equation}
    \max_{\bm{c},\bm{h}} p(\bm{f}|\bm{h},\bm{c}) 
    p(\bm{h}|\bm{c}) p(\bm{c}) = 
    \max_{\bm{c}}\left(\max_{\bm{h}} p(\bm{f}|\bm{h},\bm{c}) 
    p(\bm{h}|\bm{c})\right)p(\bm{c}).
\end{equation}
The inner one optimizes over the gray values $\bm{h}$ given a mask $\bm{c}$, 
and the outer one optimizes over the masks $\bm{c}$.
If we again assume that $\bm{f} = \bm{r}(\bm{c},\bm{h}) + \bm{n}$, where 
$\bm{n}\sim\mathcal{N}(\bm{0}, \sigma^2_{\bm{n}}\bm{I})$, then the density 
$p(\bm{f}|\bm{h}, \bm{c})$ is given by a Gaussian
\begin{equation}
    p(\bm{f}|\bm{h}, \bm{c}) = \frac{1}{(2\pi \sigma_{\bm{n}}^2)^{N/2}}
    \exp\left(-\frac{\|\bm{r}(\bm{c},\bm{h})-\bm{f}\|^2_2}
    {\sigma_{\bm{n}}^2}\right).
\end{equation}
Assuming also that the gray values are normally distributed, i.e., 
$\bm{h}|_{\bm{c}}\sim\mathcal{N}(\bm{0},\sigma^2_{\bm{h}|_{\bm{c}}}\bm{I})$, 
then the minimization problem w.r.t. $\bm{h}$ is what we call 
\emph{the regularized tonal optimization problem}:
\begin{equation}
    \argmax_{\bm{h}|_{\bar{\bm{c}}}=\bm{0}}  
    \exp\left(-\frac{\|\bm{r}(\bm{c},\bm{h})-\bm{f}\|^2_2}{\sigma_{\bm{n}}^2}
    - \frac{\|\bm{h}|_{\bm{c}}\|^2_2}{\sigma^2_{\bm{h}|_{\bm{c}}}}\right) 
    = \argmin_{\bm{h}|_{\bar{\bm{c}}}=\bm{0}} 
    \|\bm{B}_{\bm{c}}\bm{h}|_{\bm{c}} - \bm{f}\|^2_2 
    + \frac{\sigma_{\bm{n}}^2}{\sigma^2_{\bm{h}|_{\bm{c}}}} 
    \|\bm{h}_{\bm{c}}\|^2_2,
\end{equation}
where the solution is the same as in \labelcref{eq:estimate_tonal_opt}:
\begin{equation}
    \bm{h}|^*_{\bm{c}} = 
    \left(\frac{\sigma^2_{\bm{n}}}{\sigma_{\bm{h}|_{\bm{c}}}^2}\bm{I} 
    + \bm{B}_{\bm{c}}^{\T}\bm{B}_{\bm{c}}\right)^{-1}
    \bm{B}_{\bm{c}}^{\T} \bm{f}.
\end{equation}
Note that this can already be used for denoising with just a single 
inpainting with a mask $\bm{c}$, provided that we know the ratio of the 
variances of the noise and the gray values. 
The above expression suggests that we can then just apply a regularized 
tonal optimization to get the best MAP estimate. 
As before, we may take $\sigma_{\bm{h}|_{\bm{c}}}\to\infty$ to 
get classical tonal optimization if desired.
Of course, we also need to optimize w.r.t.\ the masks according to 
$p(\bm{c})$. In fact, if we take $p(\bm{c}) = 0$ for $\|\bm{c}\|_0\ne m$, 
and $p(\bm{c})$ being equal for all $\|\bm{c}\|_0=m$,
then we get the spatial optimization problem with tonally optimized  
values:
\begin{equation}
    \min_{\|\bm{c}\|_0 = m} 
    \|\bm{r}(\bm{c},\bm{h}|^*_{c}(\bm{f})) - \bm{f}\|^2_2.
\end{equation}
If we take the interpolating case, we get the classical spatial optimization 
problem~\cite{MHWT+11}:
\begin{equation}
    \min_{\|\bm{c}\|_0=m} \|\bm{r}(\bm{c},\bm{f}) - \bm{f}\|^2_2.
\end{equation}
The above further motivates using spatial optimization for denoising in 
both the interpolation and approximation cases; see \Cref{sec:adaptive}.


\subsubsection{Bayesian Interpretation}

In this subsection, we discuss how the above approaches fit in a
general Bayesian perspective, which allows for meaningful interpretations 
of the occurring probabilities. 
This is valuable as MMSE and MAP estimates rely 
on a posterior $p(\bm{w}|\bm{f})$.
Using Bayes' rule, this posterior can be rewritten as
\begin{equation}
    \label{eq:Bayes_rule}
    p(\bm{w}|\bm{f}) = \frac{p(\bm{f}|\bm{w}) p(\bm{w})}{p(\bm{f})} 
    = \frac{p(\bm{f}|\bm{w}) p(\bm{w})}
    {\int_{\mathbb{R}^N}p(\bm{f}|\bm{w}) p(\bm{w})\,d\bm{w}},
\end{equation}
where $p(\bm{w})$ is the probability density function (PDF) for 
the distribution of images $\bm{w}$ from which we assume 
$\bm{f}_r$ to originate. The likelihood 
$p(\bm{f}|\bm{w})$ is the noise PDF, which in our case is a Gaussian.
The term $p(\bm{f})$ is just a normalization constant that is irrelevant 
in practice, since it is not a function of $\bm{w}$.
This shows that the task of finding a proper posterior distribution 
corresponds to introducing an appropriate prior $p(\bm{w})$ under 
a given noise distribution $p(\bm{f}|\bm{w})$. 
This is known to be crucial for good denoising performance 
of Bayesian methods, and links our DbI framework to such approaches.


\smallskip
\paragraph{Incorporating the Inpainting Operator}

To introduce an inpainting operator $\bm{r}$ into the above model, 
we make the assumption that any $\bm{w}$ is synthesized as 
$\bm{w} = \bm{r}(\bm{c},\bm{h})$ for some mask 
$\bm{c}$ and some gray values $\bm{h}|_{\bm{c}}$.
Since now the model depends on the masks we can rewrite the PDF as
\begin{equation}
    p(\bm{w}|\bm{f}) = \sum_{\bm{c}\in\{0,1\}^N} p(\bm{w}|\bm{f},\bm{c})
    p(\bm{c}|\bm{f}),
\end{equation}
which is where the conditional mask PMF $p(\bm{c}|\bm{f})$ comes into play -- 
this is the other key ingredient for DbI along with the inpainting operator. 
We will see that this PMF allows us to introduce spatial adaptivity 
(\Cref{subsec:analytic}, \Cref{fig:spatial}) 
for operators that are otherwise not spatially adaptive. 
Finally, we can also rewrite $p(\bm{w}|\bm{f},\bm{c})$ using Bayes' rule 
in order to relate the above formulation to \labelcref{eq:Bayes_rule}:
\begin{equation}
    p(\bm{w}|\bm{f},\bm{c}) = 
    \frac{p(\bm{f}|\bm{w},\bm{c})p(\bm{w}|\bm{c})}
    {p(\bm{f}|\bm{c})} = 
    \frac{p(\bm{f}|\bm{w},\bm{c})p(\bm{w}|\bm{c})}
    {\int_{\mathbb{R}^N}p(\bm{f}|\bm{w},\bm{c})p(\bm{w}|\bm{c})\,d\bm{w}}.
\end{equation}
This provides a similar interpretation, but now we have knowledge about 
the mask. As before $p(\bm{f}|\bm{w},\bm{c})$ models 
the noise, but now $p(\bm{w}|\bm{c})$ models the distribution of 
the gray values defining $\bm{w}$ given $\bm{c}$, i.e., the 
distribution of $\bm{h}|_{\bm{c}}$. As before, the denominator 
is a normalization constant that is not practically relevant.


\smallskip
\paragraph{The Mask Posterior}

Bayes' rule allows us to explore further theoretical considerations about 
the involved mask probabilities. We can study the mask posterior 
$p(\bm{c}|\bm{f})$ in more detail, using
\begin{equation}
    p(\bm{c}|\bm{f}) = \frac{p(\bm{f}|\bm{c})p(\bm{c})}{p(\bm{f})}.
\end{equation}
Now $p(\bm{c})$ models the probability of the mask $\bm{c}$ being 
generated (irrespective of $\bm{f}$) and $p(\bm{f}|\bm{c})$ 
models some measure of the noise and image content in relation 
to the mask.
In practice, ideally the density $\bm{1}^{\T}\mathbb{E}[\bm{c}]/N$ 
should be chosen 
to be inversely proportional to the standard deviation of the noise.
Similarly if we know that the noise distribution is space-variant, 
or if we suspect that features (e.g.\ edges) are present, 
we can choose the local density of $\bm{c}$ to account for that: 
higher for more prominent edges, lower for higher noise variance.
The weight of these choices are modeled by $p(\bm{f}|\bm{c})$. 
Selecting $p(\bm{c})$ is less trivial, as it needs to match 
the mask distribution of natural images, i.e., the distribution 
of natural images from the perspective of the masks used in the 
inpainting operator. 
It is simpler to choose it based on the density, i.e., 
$p(\bm{c}) = p(\|\bm{c}\|_0/N)$, which makes it blind to spatial variations,
or to just choose it as a constant, if we have no data on it.
Note that these considerations are meant to provide a different view 
on the mask posterior and an alternative strategy on how 
to construct it.
The adaptive mask selection methods that we consider in this work directly 
induce a mask posterior $p(\bm{c}|\bm{f})$ 
and do not model $p(\bm{f}|\bm{c})$ or $p(\bm{c})$. 
They are based on strategies from the noise-free case in image inpainting and
we adapt and extend them to the noisy case.
For all the approaches that we consider, we state their induced 
PMFs $p(\bm{c}|\bm{f})$ 
(see \Cref{eq:poisson_sampling}, \Cref{thm:pmf-dens}, \Cref{app:error-diff}). 


\smallskip
\paragraph{On the Importance of the Inpainting Operator}

A crucial question is whether an inpainting operator $\bm{r}$ is
suitable for modeling natural images in a sparse and robust manner, 
such that noise can be attenuated by averaging multiple nearby 
representations of a noisy image from a lower-dimensional image manifold. 
For $\bm{r}$ being homogeneous diffusion inpainting, we know that it 
has been used successfully for image compression of natural images with 
low to medium frequencies~\cite{MBWF11}. 
Moreover, we present new results in \Cref{sec:rel-dbi-hd} that relate 
the MMSE estimate to homogeneous diffusion denoising.
The large body of literature on sparse image 
approximation and compression should provide a reasonable selection 
of good inpainting operators $\bm{r}$. 
In the current work we also consider biharmonic inpainting 
(see \Cref{subsec:exp-biharmonic}).


\smallskip
\paragraph{Interplay between the Mask PMF and Homogeneous Diffusion}

The basis vectors $\bm{B}_{\bm{c}}$ for homogeneous diffusion are 
generally low-frequent and smooth, with the local frequency depending 
on the local density of the mask points. For a constant PMF $p(\bm{c}|\bm{f})$, 
i.e., a homogeneous mask density, we get a process similar to isotropic 
homogeneous diffusion, and it is in fact approximately equivalent to it,
as we demonstrate later in 
\Cref{subsec:theoretical} and \Cref{subsec:theoretical_2D}.
As such it also shares its drawbacks, i.e., smoothing equally over image 
structures and noise. More sophisticated denoising methods such as space-variant  
diffusion allow for steering the smoothing away from image structures by 
relying on a guidance image, e.g., the gradient magnitude $|\bm{\nabla}u|$. 
Similarly, we may use the PMF $p(\bm{c}|\bm{f})$ to guide the denoising. 
One instance of a PMF that we consider is inspired by a result for 
mask selection in inpainting. Belhachmi et al.~\cite{BBBW09} have argued 
that the local density of an optimal inpainting mask $\bm{c}$ should 
be proportional to the pixelwise magnitude of the Laplacian $|\bm{L}\bm{f}|$. 
In our setting this translates to constructing a PMF 
$p$ such that $\mathbb{E}[\bm{c}|\bm{f}]\sim|\bm{L}\bm{f}|$; 
see \Cref{subsec:analytic}.


\subsubsection{Convergence}
\label{subsubsec:convergence}

A question which arises is how well the estimator $\langle\bm{u}\rangle_n$ 
approximates the 
MMSE estimate $\bm{u}^{\text{MMSE}} = \mathbb{E}[\langle\bm{u}\rangle_n|\bm{f}]$ 
as a function of 
the number of samples $n$. We consider this scaling behavior in the next 
proposition.


\smallskip
\begin{proposition}[Convergence of the DbI Estimator]
\label{thm:con-dbi} The root mean square error (RMSE) 
$\sqrt{\operatorname{MSE}(\langle\bm{u}\rangle_n,
\mathbb{E}[\langle\bm{u}\rangle_n|\bm{f}])}$ 
between the 
estimator $\langle\bm{u}\rangle_n$ and its expectation 
$\mathbb{E}[\langle\bm{u}\rangle_n|\bm{f}]$ scales 
as $O(n^{-1/2})$, where $n$ is the number of sampled masks.
\end{proposition}


\begin{proof}
We first recall that we can decompose the MSE between some estimator 
$\hat{\bm{\theta}}$ 
and some fixed parameter $\bm{\theta}$ into a variance and a bias part:
\begin{equation}
\begin{split}
\text{MSE}(\hat{\bm{\theta}}, \bm{\theta}) 
&= \mathbb{E}\bigl[\|\hat{\bm{\theta}} - \bm{\theta}\|^2_2\bigr] \\
&= \mathbb{E}\bigl[\|\hat{\bm{\theta}} - 
\mathbb{E}[\hat{\bm{\theta}}]\|^2_2\bigr] + 
\|\mathbb{E}[\hat{\bm{\theta}}] - \bm{\theta}\|^2_2 \\
&= \mathbb{V}[\hat{\bm{\theta}}] + 
\operatorname{Bias}(\hat{\bm{\theta}}, \bm{\theta})^2.
\end{split}
\end{equation}
If we consider the MSE between the estimator $\langle\bm{u}\rangle_n$ and its 
expectation $\mathbb{E}[\langle\bm{u}\rangle_n|\bm{f}]$, the bias vanishes 
and we have $\text{MSE}(\langle\bm{u}\rangle_n, 
\mathbb{E}[\langle\bm{u}\rangle_n|\bm{f}]) 
= \mathbb{V}[\langle\bm{u}\rangle_n|\bm{f}]$.
The variance $\mathbb{V}[\langle\bm{u}\rangle_n|\bm{f}]$ is given by
\begin{equation}
\mathbb{V}[\langle\bm{u}\rangle_n|\bm{f}] = 
\mathbb{V}\left[\frac{1}{n}\sum_{\ell=1}^{n}
\bm{r}(\bm{c}^\ell,\bm{f})\biggr|\bm{f}\right] = 
\frac{1}{n^2}\sum_{\ell=1}^{n}
\mathbb{V}\left[\bm{r}(\bm{c},\bm{f})|\bm{f}\right] = 
\frac{1}{n}\mathbb{V}[\bm{r}(\bm{c},\bm{f})|\bm{f}].
\end{equation}
The second equality holds because the masks are independent and 
identically distributed.
For a finite variance $\mathbb{V}[\bm{r}(\bm{c},\bm{f})|\bm{f}]$, the root 
mean square error between the estimator and its expectation thus scales 
as $O(n^{-1/2})$. 
\end{proof}


\subsubsection{Acceleration by Low-Discrepancy Sequences}
\label{subsubsec:low-discrepancy}

When the masks are random variables, as noted in \Cref{subsubsec:convergence}, 
we have a somewhat slow convergence of $O(n^{-1/2})$. 
Informally this means that to decrease the RMSE by a factor 4
we would need 16 times as many samples. 
The natural question arises whether we 
can do better by trading randomness for a more structured sampling strategy. 
The answer is positive, as in the context of integration (and our problem can 
be framed as such w.r.t.\ the counting measure), a prominent approach for 
speeding up convergence is the use of low-discrepancy sequences. 
These sequences fill up space more uniformly than random sequences. 
The uniformity is typically quantified using the (star) discrepancy of the 
sequence. Theoretically, the Koksma-Hlawka inequality~\cite{KN05} allows one 
to bound the numerical integration error, i.e., 
$\|\langle\bm{u}\rangle_n-\mathbb{E}[\langle\bm{u}\rangle_n|\bm{f}]\|_2$ 
in our case, by using the product of the 
discrepancy of the sequence and the variation of the integrand. 
In practice this usually translates to a convergence that can reach 
as high as $O(n^{-1})$ which is much better than the $O(n^{-1/2})$ 
convergence for the purely random case. Experimental results illustrating 
a boost to the convergence in the DbI setting are presented in 
\Cref{subsec:exp-convergence}.


\section{Linking Denoising by Inpainting to Homogeneous Diffusion}
\label{sec:rel-dbi-hd}

The simplest approaches for mask selection in the DbI framework 
are those, that are independent of the image that is to be filtered 
($p(\bm{c}|\bm{f}) \equiv p(\bm{c})$).
We consider shifted regular masks as well as randomly selected masks.
They are characterized by a spatially flat expectation 
$\mathbb{E}[\bm{c}] = \text{const}$.
In the following, we briefly introduce regular masks, show how they can 
be used in the DbI framework and discuss the resulting 
filtering behavior. 
Then we derive relations between DbI with regular masks and homogeneous 
diffusion filtering in 1D. 
Afterwards, using random masks instead of regular masks, 
we empirically extend those results to the 2-D setting.


\subsection{Regular Masks}
\label{subsec:regular}

Regular masks are created by generating a pattern
with each $r$-th pixel in $x$- and each $s$-th pixel in $y$-direction being 
added to the mask. 
We can then shift such a mask in both directions to obtain multiple masks.
If we assume an $n_x \times n_y$ pixel grid, we can create such a regular 
mask via
\begin{equation}
c_{i, j} = \begin{cases}
1 & \text{if } i \bmod r = 0 \text{ and } j \bmod s = 0, \\
0 & \text{else.}
\end{cases}
\end{equation}
We have $s$ options of shifting this regular mask in $x$-direction and $r$ 
options in $y$-direction, adding up to $n=rs$ total possible configurations.
Denoting by $p \in \{0, \dots, r-1\}$ and $q \in \{0, \dots, s-1\}$ the shift 
in $x$- and $y$-direction, respectively, we can write the shifted masks as 
\begin{equation}
c^{ps+q+1}_{i, j} = \begin{cases}
1 & \text{if } i \bmod r = p \text{ and } j \bmod s = q, \\
0 & \text{else.}
\end{cases}
\end{equation}
Clearly, the created masks are independent of the image. 
Furthermore, the mask density is constant over the entire image, leading to 
the same smoothing strength at all locations, solely determined by the total 
mask density, i.e., by the spacing. If $r=s$, this smoothing is 
equally strong in $x$- and $y$-direction.
Visually one then observes a smoothing behavior that resembles the one of 
homogeneous diffusion filtering 
(see \Cref{fig:regular}(c) and \Cref{fig:regular}(d)).
The influence of the mask density on the smoothing strength can be observed 
in \Cref{fig:regular}(d) and \Cref{fig:regular}(e).


\begin{figure}[tbhp]
\centering
\tabcolsep2pt
\begin{tabular}{ccccc}
\textbf{(a) original} & \textbf{(b) noisy} & 
\textbf{(c)	HD} & \textbf{(d) DbI-R} & \textbf{(e) DbI-R} \\[1.2mm]
\includegraphics[width=0.185\linewidth]{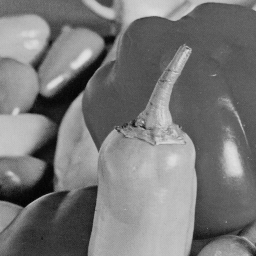} & 
\includegraphics[width=0.185\linewidth]{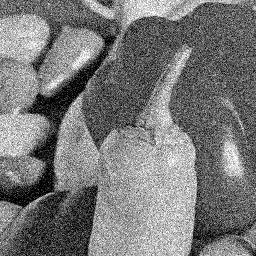} & 
\includegraphics[width=0.185\linewidth]
{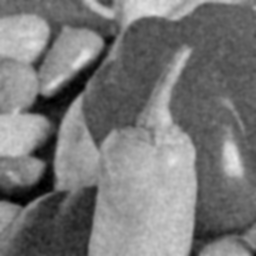} & 
\includegraphics[width=0.185\linewidth]
{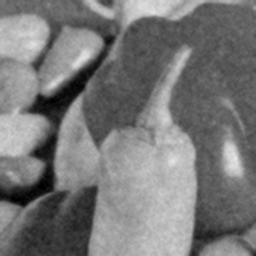} & 
\includegraphics[width=0.185\linewidth]
{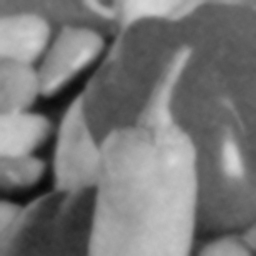} \\[-0.8mm]
\emph{peppers} & $\sigma_n=30$ & $T=1.35$ & $r=s=3$ & $r=s=5$
\end{tabular}
\caption{\label{fig:regular} Comparison of homogeneous diffusion (HD) 
and denoising by inpainting with regular masks (DbI-R) 
on the test image \emph{peppers} with $\sigma_n=30$. 
\Cref{fig:regular}(c) and \Cref{fig:regular}(d) show the visual 
similarities of both methods. 
\Cref{fig:regular}(d) and \Cref{fig:regular}(e) 
illustrate the influence of the expected density 
$\bm{1}^{\T}\mathbb{E}[\bm{c}]/N$ on the smoothness of the 
reconstruction: \Cref{fig:regular}(e) was intentionally chosen 
with a density that is too low resulting in too much smoothing.}
\end{figure}


The similarity between the methods can not only be observed visually, but also 
established theoretically: 
Next we provide a derivation in the 1-D case for regular 
masks relating the diffusion time of homogeneous diffusion 
to the mask density in DbI. 


\subsection{Mathematical Analysis in 1D}
\label{subsec:theoretical}

We consider a discrete 1-D signal $\bm{f}$ and regular inpainting masks 
with spacing $r$ and shift $p \in \{0, \dots, r-1\}$.
It is known that in 1D, homogeneous diffusion inpainting and linear 
interpolation are equivalent. 
Thus, an inpainted pixel at position $i$ can be described in terms of its two 
neighboring mask pixels. We denote the distance between the pixel $i$ and its 
neighboring mask pixel on the left by $\ell \coloneqq |i-p| \bmod r$, 
which implies that for mask pixels we have $\ell=0$. Accordingly, the distance
to the mask pixel on the right is given by $r-\ell$.
The interpolated value at pixel $i$ for mask $\ell+1$ is then
\begin{equation}
v_i^{\ell+1} = \frac{r-\ell}{r} f_{i-\ell} + \frac{\ell}{r} f_{i+r-\ell}.
\end{equation}
To obtain the final result, the inpaintings from the $r$ shifted masks are 
averaged. We get
\begin{equation}
\begin{split}
\label{eq:regular-mask}
u_i = \frac{1}{r} \sum_{\ell=0}^{r-1} v_i^{\ell+1} &= 
\frac{1}{r} \left(f_i + \sum_{\ell=1}^{r-1} \frac{r-\ell}{r} f_{i-\ell} + 
\frac{\ell}{r} f_{i+r-\ell}\right) \\
&= \frac{1}{r^2} \left(r f_i + \sum_{\ell=1}^{r-1} \ell \left(f_{i-(r-\ell)} +
f_{i+(r-\ell)}\right)\right),
\end{split}
\end{equation}
where the last line reveals the general form of the filter in dependence of 
the spacing $r$: The filter is given by a hat kernel with central weight 
$1/r$ and width $2r-1$. In \Cref{thm:mddt} we demonstrate that this kernel 
can be seen as a consistent discretization of 
$\partial_t u = \partial_{xx}u$. 
Consequently, convolution with such a kernel approximates Gaussian 
smoothing, which explains the visual similarity of the results 
in \Cref{fig:regular}.
Since the spacing $r$ determines the size of the smoothing kernel, 
we explicitly see the connection between the mask density and the 
smoothing strength.
For the special case of $r=2$, \labelcref{eq:regular-mask} yields
\begin{equation}
\label{eq:standard_explicit_hd}
u_i = \frac{f_{i-1} + 2 f_i + f_{i+1}}{4},
\end{equation}
which is exactly a single step of an explicit scheme for homogeneous 
diffusion with step size $T = \frac{1}{4}$ and initial signal $\bm{f}$ 
(assuming grid size $h=1$).
If we reformulate \labelcref{eq:regular-mask} in a way that resembles an 
explicit scheme for homogeneous diffusion, we can derive a general connection 
between the spacing $r$ (and thus the density) of denoising by inpainting 
with regular masks and the time step size of such an explicit scheme, which
we state in \Cref{thm:mddt}.


\smallskip
\begin{theorem}[Connection between Mask Density and Diffusion Time]
\label{thm:mddt} Given the $r$ shifted regular inpainting 
masks in 1D, each of density $d = 1/r$, denoising by inpainting 
approximates explicit homogeneous diffusion at time 
\begin{equation}
    \boxed{T = \frac{1-d^2}{12 d^2}}
\end{equation}
\end{theorem}


\begin{proof}
In \labelcref{eq:regular-mask} we derived the general form of the filter 
corresponding to denoising by inpainting with regular masks of spacing $r$ as 
\begin{equation}
u_i = \frac{1}{r^2} \left(r f_i + \sum_{\ell=1}^{r-1} \ell 
\left(f_{i-(r-\ell)} + f_{i+(r-\ell)}\right) \right).
\end{equation}
We can rewrite this in the following manner:
\begin{equation}
\begin{split}
u_i &= \frac{1}{r^2} \left(r f_i + \sum_{\ell=1}^{r-1} \ell 
\left(f_{i-(r-\ell)} + f_{i+(r-\ell)}\right) \right) \\
&= \frac{1}{r^2} \left( r^2 f_i -  2 \sum_{\ell=1}^{r-1} \ell f_i + 
\sum_{\ell=1}^{r-1} \ell \left(f_{i-(r-\ell)} 
+ f_{i+(r-\ell)} \right) \right) \\
&= f_i + \frac{1}{r^2} \sum_{\ell=1}^{r-1} \ell 
\left(f_{i-(r-\ell)} - 2 f_i + f_{i+(r-\ell)} \right),
\end{split}
\end{equation}
where we have used that $\sum_{\ell=1}^{r-1} \ell = \frac{(r-1)r}{2}$. 
Then we may write
\begin{equation}
\begin{split}
u_i - f_i &= 
\frac{1}{r^2} \sum_{\ell=1}^{r-1} \ell 
\left( f_{i-(r-\ell)} - 2 f_i + f_{i+(r-\ell)} \right) \\
&= \frac{1}{r^2} \sum_{\ell=1}^{r-1} \ell (r-\ell)^2 \, 
\frac{ f_{i-(r-\ell)} - 2 f_i + f_{i+(r-\ell)}}{(r-\ell)^2} \\
&= \sum_{\ell=1}^{r-1}  \frac{\ell (r-\ell)^2}{r^2}  \, 
\frac{ f_{i-(r-\ell)} - 2 f_i + f_{i+(r-\ell)}}{(r-\ell)^2}.
\end{split}
\end{equation}
By approximating $f_{i \pm (r-\ell)}$ via a Taylor expansion and using 
the sampling distance $h$, we can derive the time step size as
\begin{equation}
\begin{split}
u_i - f_i 
&= \sum_{\ell=1}^{r-1}  \frac{\ell (r-\ell)^2}{r^2}  \, 
\frac{ f_{i-(r-\ell)} - 2 f_i + f_{i+(r-\ell)}}{(r-\ell)^2} \\
&= \sum_{\ell=1}^{r-1} \left(\frac{\ell (r-\ell)^2}{r^2} \right) 
\left(h^2 \left. d_{xx} f  \right|_{i} + \frac{h^4 (r-\ell)^2}{12} 
\left. d_{xxxx} f \right|_{i} + \mathcal{O}(h^6) \right) \\
&= h^2 \sum_{\ell=1}^{r-1} \left(\frac{\ell (r-\ell)^2}{r^2} \right) 
\left( \left. d_{xx} f  \right|_{i} +  \mathcal{O}(h^2) \right) \\
&\approx h^2 \sum_{\ell=1}^{r-1} \left(\frac{\ell (r-\ell)^2}{r^2} \right) 
\left. d_{xx} f  \right|_{i}.
\end{split}
\end{equation}
We end up with an approximation of an explicit scheme with time step size
\begin{equation}
T = h^2 \sum_{\ell=1}^{r-1} \frac{\ell (r-\ell)^2}{r^2} 
= \frac{h^2 (r^2-1)}{12}.
\end{equation}
Using that the density is the inverse of the grid spacing and setting
$h=1$, we derive the final relation between $T$ and the density $d$, given by 
\begin{equation}
T = \frac{1-d^2}{12 d^2}.
\end{equation}
\end{proof}


\subsection{Empirical Extension to 2D}
\label{subsec:theoretical_2D}

To derive the relationship to the diffusion time in the 1-D 
case we used the fact that the solution of the Laplace equation 
with Dirichlet boundaries is given by linear interpolation. 
That is, we know the closed form of the inpainting echoes in 1D. 
In 2D a closed form solution for those is not known, however they 
may be computed numerically. Thus our goal is to establish a 
relationship between the diffusion time and the density numerically.

We take as a starting point the ansatz from the 1-D case 
that the diffusion time $T$ is given as $\frac{1-d^2}{12d^2}$, 
but generalize it to the form 
$T\approx \frac{1-d^{\gamma}}{\beta d^{\gamma}}$. 
Provided that this conjecture is correct we only need to find 
the constants $\beta$ and $\gamma$.
Since regular masks only allow for a stepwise adaptation of 
the mask density, they are not well-suited for generating a 
large number of data points at different densities. Therefore, 
we use uniform random masks instead, which also have a spatially flat 
expectation, i.e., $\mathbb{E}[\bm{c}] = \text{const}$. 

First we numerically tabulate the relationship between the 
density and the diffusion time. That is, given a density $d$ 
we find the diffusion time $T(d)$ which minimizes the difference 
between the filter matrices:
\begin{equation}
    T(d) = \underset{T\geq 0}{\argmin}
    \,\|\bm{A}_{DbI}(d) - \bm{A}_{HD}(T)\|^2_F.
\end{equation}
Here $\|\cdot\|_F$ is the Frobenius norm, and the matrices 
are the DbI filter matrix resulting from a probability mass 
function for masks with expected density $d$, and the matrix
modeling homogeneous diffusion at time $T$ using 
an implicit Euler discretization:
\begin{align}
    \bm{A}_{DbI} (d) &\coloneqq \mathbb{E}\left[\left(\bm{C}+
\left(\bm{I}-\bm{C}\right)\bm{L}\right)^{-1}\bm{C}\right], 
\quad \frac{1}{N}\bm{1}^{\T}\mathbb{E}[\bm{c}] = d, \\
\bm{A}_{HD}(T) &\coloneqq \left(\bm{I}+T \bm{L}\right)^{-1}.
\end{align}
We estimate $\bm{A}_{DbI}$ using 1024 sampled masks.
Then, having the relationship $d \mapsto T(d)$ we find that 
$T(d) \approx \frac{1-d^{\gamma}}{\beta d^{\gamma}}$ for 
$\beta = 4.58$, $\gamma = 1.3$, which is illustrated in 
\Cref{fig:2d}. Note the high quality of the data fit, which 
confirms the accuracy of the derived relation.


\begin{figure}[tbhp]
\centering
\tabcolsep1pt
\begin{tabular}{cc}
\includegraphics[width=0.48\linewidth]
{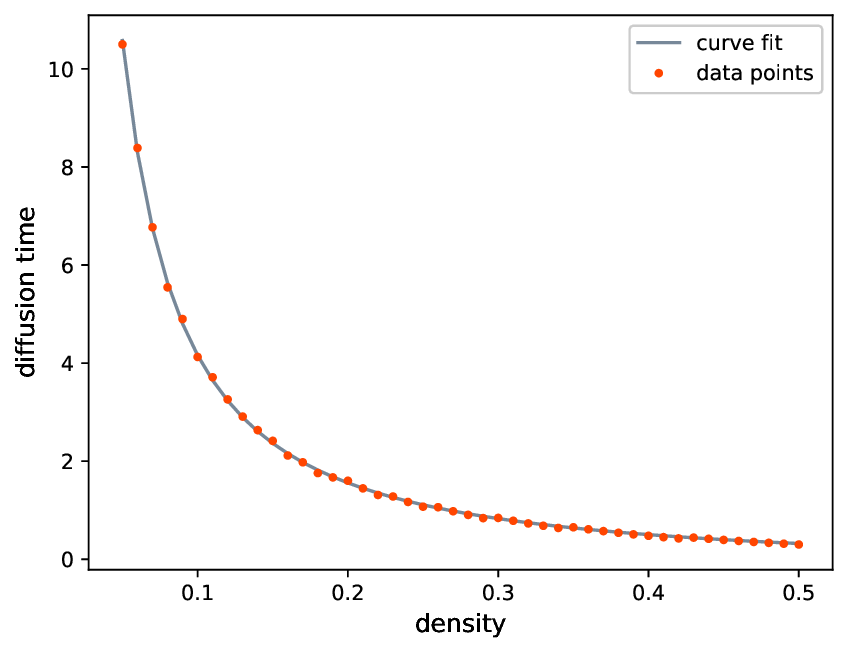}
\end{tabular}
\caption{\label{fig:2d} The fit based on the ansatz 
$\frac{1-d^{\gamma}}{\beta d^{\gamma}}$ with $\beta=4.58$, 
$\gamma = 1.3$ and the tabulated correspondence 
between density and diffusion time. The results are obtained with
denoising by inpainting with uniform random masks and an implicit scheme
for homogeneous diffusion.
They show that also in 2D our ansatz captures the relation between 
mask density and diffusion time very accurately.}
\end{figure}


\section{Spatial Optimization for Denoising by Inpainting}
\label{sec:adaptive}

As we have seen in \Cref{sec:rel-dbi-hd}, the use of non-adaptive masks 
restricts the DbI framework, as it entails a non-adaptive smoothing behavior.
Furthermore, our results from \Cref{subsec:prob-theory} emphasize the 
importance of spatial optimization in the context of image denoising.
In~\cite{APW17}, an adaptive mask selection approach enables the framework to
perform edge-preserving image filtering, although the simple homogeneous  
diffusion inpainting operator by itself is space-invariant. 
This approach thus implies a different paradigm for image denoising:
\emph{Instead of optimizing the denoising operator, one can optimize the
data.} In this section, we will first present the strategy that was proposed
in~\cite{APW17}. Then we propose an alternative, simpler approach that 
eventually gains its power by the application of tonal optimization.


\subsection{Densification Method}
\label{subsec:probabilistic}

Two well-known mask selection strategies from image compression are 
probabilistic sparsification \cite{MHWT+11} and densification \cite{HMWP13}, 
which build the mask in an iterative way using a top-down and 
a bottom-up strategy, respectively.

In probabilistic sparsification, we start with a full mask and take away 
the least important pixels from a number of randomly selected candidates 
in each iteration.
To identify those pixels, we temporarily exclude all candidates from the 
mask and compute an inpainting. Then the candidate locations with the 
highest local (i.e., pixelwise) reconstruction error are added back to the 
mask as they are assumed to be the most important, while the others remain
permanently excluded. This process is repeated until the desired mask 
density is reached.
In probabilistic densification, the initial mask is empty and again a number 
of candidate pixels are selected. Given an inpainting with the current mask 
(in the first step some pixels have to be chosen at random) 
we select and add those candidates to the mask that have the highest 
local reconstruction error.

In the noisy setting, special care is required as the pixel selection 
based on the \emph{local} reconstruction error is not reliable. 
The local error does not allow the algorithm to distinguish 
between noise and important image structures, such as edges. If a
pixel contains strong noise, this creates a large local error 
because -- just like edges -- the noise cannot be 
reconstructed by the smooth inpainting. Introducing 
such a noisy pixel into the mask is not desirable.
We cure this problem by judging the importance 
of a pixel based on its effect on the \emph{global} 
reconstruction error. We do this by calculating a 
full inpainting for each candidate pixel. While this improves 
the quality of the selected mask, it drastically increases the run time.

Even though in the noise-free setting densification and sparsification yield 
results of comparable quality \cite{Ho16}, this is different when 
handling noisy data. 
For sparsification we initially have very dense masks. If we exclude candidate 
pixels from such masks, the reconstructions often only differ at the locations 
of these pixels. Therefore, sparsification tends to keep noisy pixels in 
the mask, even when a global reconstruction error is computed. 
This problem does not occur in probabilistic densification, as for a sparse 
mask, the candidate pixels have a global influence. 
The result of this effect is illustrated in \Cref{fig:densification}. 
Here, densification is able to select appropriate pixels that lead to an 
almost perfect result while sparsification fails to reconstruct the image 
properly.

Thus, we opt for a \emph{probabilistic densification} algorithm based on a 
\emph{global} error computation, which is described in 
\Cref{alg:densification} and has been proposed in \cite{APW17}.
An additional advantage of this probabilistic densification method is that 
it does not only select pixels at useful locations (e.g., close to edges), 
but also implicitly avoids picking pixels that are too noisy, as they would 
have a negative impact on the reconstruction quality.
The method can be interpreted according to the probabilistic mask generation 
framework from \Cref{sec:framework}, and we provide the implied mask 
probabilities in the following proposition.


\smallskip
\begin{proposition}[Mask Probabilities implied by the Densification
Method]
\label{thm:pmf-dens}
A mask $\bm{c}$ generated by probabilistic densification has the 
conditional probability density function
\begin{equation}
    p(\bm{c}|\bm{f}) = 
    \sum_{\sigma \in S_m}p_{\sigma}(\bm{c}|\bm{f}),
\end{equation}
where $m=\|\bm{c}\|_0$ is the number of mask pixels, 
the sum is taken over the group $S_{m}$ of permutations of 
the ordering of the $m$ mask pixels, and $p_{\sigma}(\bm{c}|\bm{f})$ 
denotes the probability that the $m$ mask points were introduced 
in the order $\sigma$. The latter is the product of the probabilities 
of selecting one mask pixel at each step:
\begin{equation}
    p_{\sigma}(\bm{c}|\bm{f}) = p^m_{\sigma}(\bm{c}|\bm{f}) 
    \ldots p^1_{\sigma}(\bm{c}|\bm{f}).
\end{equation}
The probability of picking the $k$-th mask pixel (according to 
the permutation $\sigma$) at step $k$ has the following form:
\begin{equation}
    p^k_{\sigma}(\bm{c}|\bm{f}) = 
    \sum_{\beta=1}^{\alpha}\frac{1}{\beta}
    \frac{\binom{N_{eq}-1}{\beta-1} \binom{N_{gt}}{\alpha-\beta}}
{\binom{N-k}{\alpha}},
\end{equation}
where $\alpha$ is the number of candidates considered per step, 
$N_{gt}$ is the number of non-mask pixels at step $k$ that would have 
resulted in an inpainting with a higher MSE if they were chosen instead of the 
$k$-th mask pixel in $\sigma$, and $N_{eq}$ is the number of non-mask pixels 
that would have resulted in the same MSE.

\end{proposition}


\begin{proof}
We present the proof of this result in \Cref{app:densification}.
\end{proof}


\begin{figure}[tbhp]
\centering
\tabcolsep10pt
\fboxsep=0.5pt
\begin{tabular}{ccc}
\textbf{(a) input} & \textbf{(b) sparsification} & 
\textbf{(c)	densification} \\[1.2mm]
\includegraphics[width=0.24\linewidth]
{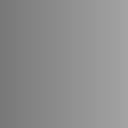} & 	
\includegraphics[width=0.24\linewidth]
{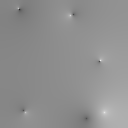} &
\includegraphics[width=0.24\linewidth]
{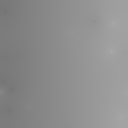} \\[-0.8mm]
original & MSE: $76.07$  & MSE: $1.98$ \\[1.2mm]
\includegraphics[width=0.24\linewidth]
{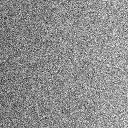} & 
\fbox{\includegraphics[width=0.24\linewidth]
{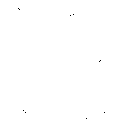}} &
\fbox{\includegraphics[width=0.24\linewidth]
{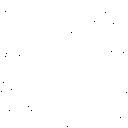}} \\[-0.8mm]
noisy ($\sigma_n=30$) & optimized mask & optimized mask
\end{tabular}
\caption{\label{fig:densification} Comparison of sparsification 
and densification on a synthetic test image with 
$\sigma_n=30$~\cite{APW17}.
For both methods, the mask density $d$ was optimized with a grid search 
w.r.t. the MSE. 
The noisy gradient image is not reconstructed adequately by sparsification, 
since it favors keeping noisy pixels in the first iterations due to 
localization. Densification does not suffer from this problem and thereby 
achieves a better denoised reconstruction.}
\end{figure}


\begin{algorithm}[tbhp]
\caption{Mask densification with global error computation \cite{APW17}.
\label{alg:densification}}
\small
\begin{description}
\hrule\vspace{2mm}
\item[\emph{Input:}]
Noisy image $\bm f \in \R^{N}$, number of candidates $\alpha$, desired 
final mask density $d$.
\item[\emph{Initialization:}]
Mask $\bm c=0$ is empty.
\item[\emph{Compute:}]\hfill\\
\textbf{do} 
\begin{enumerate}
\item[1.] Choose randomly a set $A \subset \{k \in \{1,...,N\} \, | \, 
c_k = 0\}$  with $\alpha$ candidates.\\
\textbf{for all} $i \in A$ \textbf{do} 
\begin{enumerate}
\item[2.] Set temporary mask $\bm m^{i}$ such that $\forall k \in 
\{1,...,\alpha\} \setminus \{i\}: m^{i}_k = c_k$, $m^{i}_i=1$.
\item[3.] Compute reconstruction $\bm u^i$ from mask $\bm m^{i}$ 
and image data $\bm f$.
\end{enumerate}
\textbf{end for}	
\item[4.] Set $\bm c = \underset{\bm m^{i}, \, i \, \in A}{\argmin}
\,\textnormal{MSE}(\bm u^{i}, \bm f)$. This adds one mask point to $\bm c$.
\end{enumerate}
\textbf{while} pixel density of $\bm c$ smaller than $d$.
\item[\emph{Output:}]
Mask $\bm c$ of density $d$.\\
\vspace*{-2mm}\hrule
\end{description}
\end{algorithm}


\subsection{Acceleration via the Analytic Results of Belhachmi et al.}
\label{subsec:analytic}

As the global error computation in the previous approach requires calculating 
an inpainting for each candidate pixel, the run time is substantial.
Therefore, we propose another approach, with the goal of a faster mask 
generation process. We refer to this method as the \emph{analytic
method}. It is based on the results of Belhachmi et al.~\cite{BBBW09}.
They have shown that the mask density for 
homogeneous diffusion inpainting should be proportional to the pointwise 
magnitude of the Laplacian $|\bm{L}\bm{f}|$. 
Additionally, they suggest using the Gaussian-smoothed version 
$\bm{f}_{\sigma} \coloneqq \bm{K}_{\sigma} \ast \bm{f}$ of $\bm{f}$ even in 
the noise-free setting. Here $\bm{K}_{\sigma}$ is a discrete approximation of 
a Gaussian with standard deviation $\sigma$. 
This step proves even more beneficial in our setting, since we are calculating 
the Laplacian of noisy data, and regularizing $\bm{f}$ helps considerably 
for constructing a reasonable guidance image $|\bm{L}\bm{f}_{\sigma}|$.

As we require multiple different binary masks for our framework, we sample 
from $|\bm{L} \bm{f}_{\sigma}|$ by using a simple and fast Poisson sampling.
Given a density image $\bm{d}\in [0,1]^N$, we can sample a mask according to 
it by generating a uniform random number $v_i\sim U[0,1]$ for each pixel 
$i$ and then thresholding at $d_i$:
\begin{equation}
c_i = \begin{cases} 1 & 
\text{if } v_i \leq d_i,\\ 0 &
\text{if } v_i > d_i.\end{cases}
\end{equation}
Then the probability mass function $p_{\bm{d}}$ for sampling a mask $\bm{c}$ 
given the density image $\bm{d}$ is
\begin{equation}
\label{eq:poisson_sampling}
p_{\bm{d}}(\bm{c}) = \frac{1}{P}\prod_{i=1}^N (d_i)^{c_i}(1-d_i)^{1-c_i}, 
\quad P = \sum_{\bm{c}\in\{0,1\}^N}\prod_{i=1}^N (d_i)^{c_i}(1-d_i)^{1-c_i}.
\end{equation}
By construction the mask would have an expected density equal to the mean 
value of $\bm{d}$. 
In our approach we set the per pixel probabilities to 
\begin{equation}
\bm{d} = \min\left\{C|\bm{L}\bm{f}_{\sigma}|,\bm{1}\right\},
\label{eq:belhachmi}
\end{equation}
where the minima are taken pointwise, 
and $C$ is a constant chosen such that the mean value of $\bm{d}$ is equal 
to the desired mask density. 
\Cref{fig:pipeline} shows the pipeline for mask generation with this method. 
One can observe in \Cref{fig:pipeline}(b) that $\bm{d}$ is strongly affected 
by the noise despite the pre-smoothing. 
This is because we calculate second-order derivatives that are even more 
sensitive to noise. When sampling from this image the mask is drawn towards 
noisy pixels. 
To counteract this, we propose to perform an additional outer smoothing of 
the probability image $\bm{d}$, after the absolute value of the Laplacian is 
taken, thus modifying it to 
\begin{equation}
\bm{d} = \min\left\{C \left(\bm{K}_{\rho} \ast |\bm{L}\bm{f}_{\sigma}|\right), 
\bm{1}\right\},
\label{eq:belhachmi-postsm}
\end{equation}
with a post-smoothing parameter $\rho$.
Our proposed selection strategy offers an instant generation of adaptive 
masks, in a sense that it does not require the calculation of any inpainting. 
Furthermore, it provides a transparent formulation of the mask PMF
(see \labelcref{eq:poisson_sampling}) and as such exhibits a specifically 
simple interpretation in the context of our probabilistic framework in
\Cref{subsec:prob-theory}.
On the other hand, contrary to probabilistic densification it does not have 
a mechanism to avoid noisy mask pixels.
To obtain the best possible results, the pre-smoothing parameter $\sigma$, 
the post-smoothing parameter $\rho$, and the desired mask density have to be 
optimized depending on the image content and the noise level.

Note that Belhachmi et al.~\cite{BBBW09} apply Floyd-Steinberg 
dithering~\cite{FS76}, which includes an error diffusion in the binarization 
process.
This strategy can be equipped with a random component in order to generate 
multiple masks, which makes it an alternative to Poisson sampling for us.
We have tested both methods and found that there is no advantage in using 
Floyd-Steinberg dithering. Thus, we opt for the simple Poisson sampling.
Nonetheless, we give the mask probabilities for sampling with 
error diffusion methods in \Cref{app:error-diff}.


\begin{figure}[tbhp]
\centering
\includegraphics[width=\textwidth]{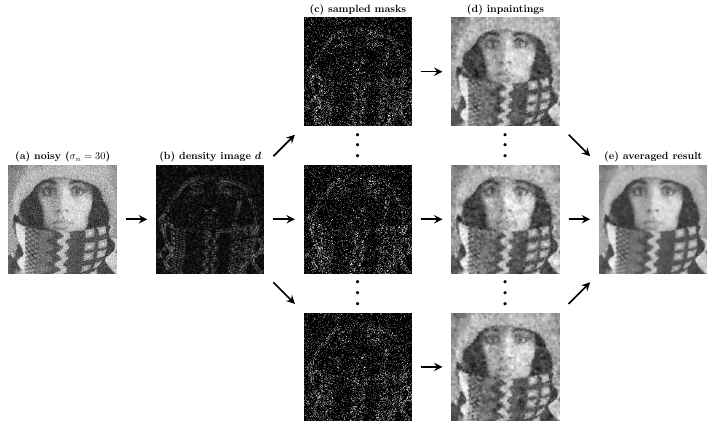} 
\caption{\label{fig:pipeline}
Pipeline for mask generation with the analytic method. 
(a) Test image \emph{trui} with $\sigma_n=30$. 
(b) Target image (without post-smoothing) from which masks are sampled. 
(c) Three examples of Poisson-sampled masks. 
(d) Corresponding homogeneous diffusion inpaintings. 
(e) Averaged inpaintings (from 32 masks), final denoising result.}
\end{figure}


\section{Experiments}
\label{sec:experiments}

In this section, we present our experiments. They evaluate our theories 
and compare the different DbI strategies in practice.
Firstly, we confirm the accuracy of the 1-D relation that we derived for
DbI with regular mask in \Cref{subsec:theoretical}. 
We also display the corresponding results in 2D.
Next, we show that the theoretical convergence estimates 
from~\Cref{subsubsec:convergence} also hold in practice and evaluate the
gain through low-discrepancy-based sampling 
(see \Cref{subsubsec:low-discrepancy}).
Furthermore, we assess the spatial and tonal mask optimization approaches.
To this end, we compare DbI to PDE-based methods of similar structural 
complexity.
Aside from homogeneous diffusion, we choose linear space-variant diffusion 
and nonlinear diffusion as representatives of methods that are based on 
operator optimization.
Lastly, we consider the denoising by biharmonic inpainting to further
investigate the question of data optimization vs.\ operator optimization.


\subsection{Relation Between DbI and Homogeneous Diffusion}
\label{subsec:exp-relation-dbi-hd}

In \Cref{subsec:theoretical} we derived a relation between the mask density 
$d$ and the diffusion time $T$, given by $T = (1-d^2)/(12 d^2)$.
To confirm that this relation allows for a good estimate of the diffusion 
time in practice, we perform an experiment on a 1-D signal, 
which is generated by extracting the 128th row of the \emph{peppers} 
test image.
Homogeneous diffusion is implemented using explicit Euler and the spatial 
discretization from \labelcref{eq:standard_explicit_hd} with the number of 
iterations chosen such that the desired diffusion time $T$ is reached.
The result in \Cref{fig:1d} demonstrates that the diffusion time 
obtained via \Cref{thm:mddt} is a good approximation.


\begin{figure}[tbhp]
\centering
\tabcolsep1pt
\begin{tabular}{cc}
\textbf{(a) $r=5, T=2.0$} & \textbf{(b) $r=10, T=8.25$} \\
\includegraphics[width=0.48\linewidth]{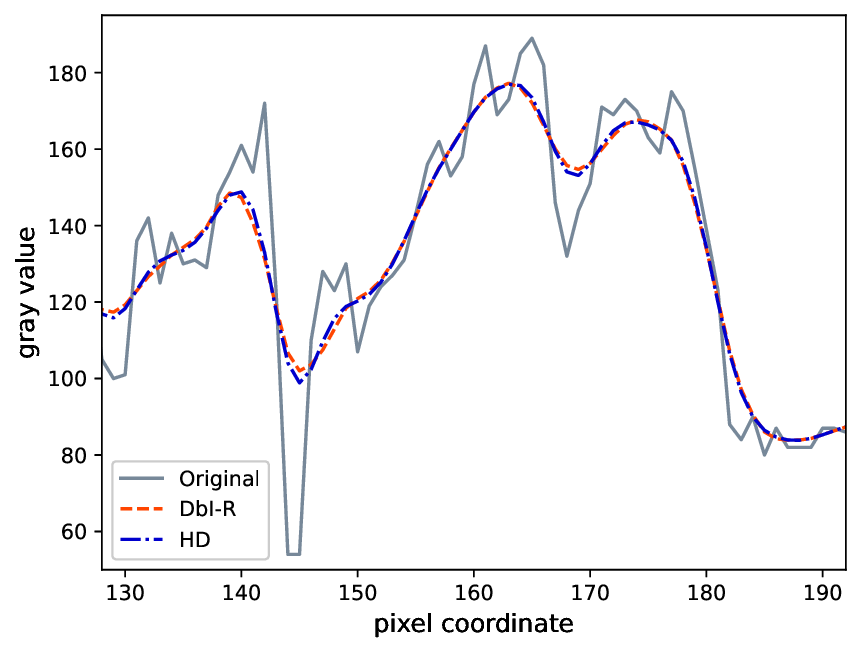} & 
\includegraphics[width=0.48\linewidth]{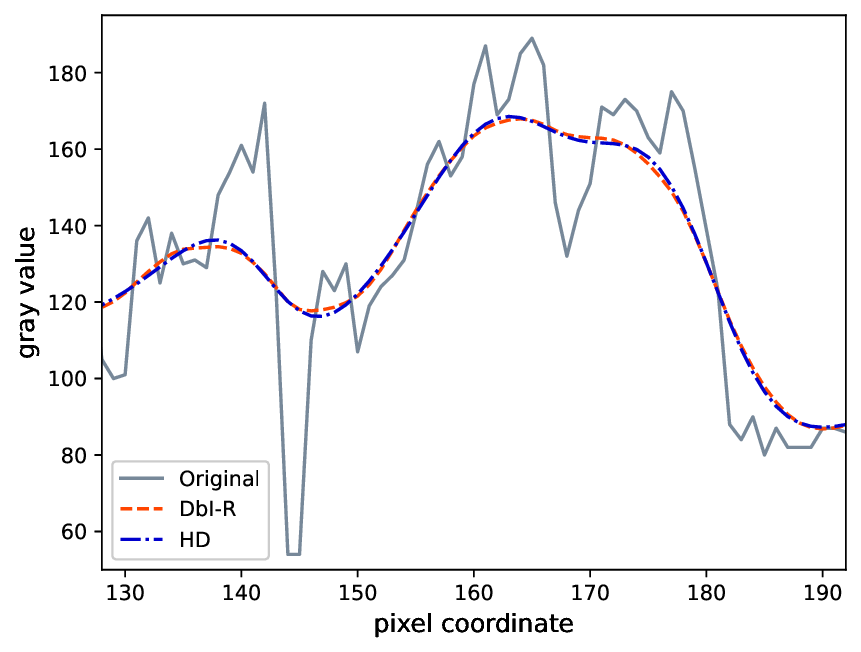}
\end{tabular}
\caption{\label{fig:1d} Comparison of denoising by inpainting with shifted 
regular masks (DbI-R) and homogeneous diffusion (HD) on a 
one-dimensional signal (128th row of the test image \emph{peppers}).
We display a section from the original signal and filtered versions 
obtained with denoising by inpainting with regular masks of 
spacing $r$ and homogeneous diffusion filtering with diffusion time 
$T$, calculated according to \Cref{thm:mddt}. 
We see that both filters lead to very similar results, confirming 
that the approximation from the theorem is indeed realistic.}
\end{figure}


In \Cref{subsec:theoretical_2D} we extended this relation to 2D, yielding
$T = (1-d^{\gamma})/(\beta d^{\gamma})$ with $\beta=4.58$ and $\gamma = 1.3$.
To confirm this, we now consider the 2-D \emph{peppers} test image.
We perform denoising by inpainting with 1024 randomly selected masks, 
as well as homogeneous diffusion filtering with the diffusion time 
calculated according to the above relation and compare the results. 
The experiments in \Cref{fig:theoretical-2d} visually and qualitatively 
confirm the accuracy of the relation in 2D.


\begin{figure}[tbhp]
\centering
\tabcolsep5pt
\begin{center}
\begin{tabular}{cc|cc}
\textbf{(a) $d=5\%$} & \textbf{(b) $T \approx 10.5$} 
& \textbf{(c) $d=1\%$} & \textbf{(d) $T \approx 86.7$} \\[1.2mm]
\includegraphics[width=0.22\textwidth]
{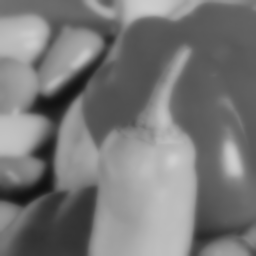} &
\includegraphics[width=0.22\textwidth]
{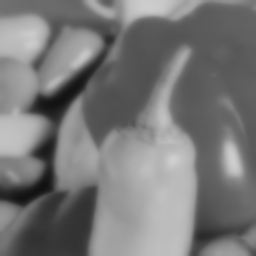} &
\includegraphics[width=0.22\textwidth]
{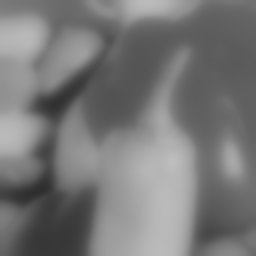} &
\includegraphics[width=0.22\textwidth]
{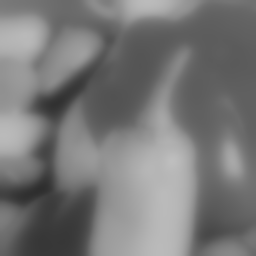} \\[-0.8mm]
DbI-Ran & HD Filtering & DbI-Ran & HD Filtering
\end{tabular}
\end{center}
\caption{\label{fig:theoretical-2d} 
Comparison of denoising by inpainting with 1024 random masks (DbI-Ran) and 
homogeneous diffusion (HD) on the test image \emph{peppers}.
The diffusion times $T$ corresponding to the mask densities $d$ are calculated
according to the result from \Cref{subsec:theoretical_2D}.
The MSE between (a) and (b) is 0.61 and the MSE between (c) and (d) is 6.37.
This shows that the empirically derived relation is accurate, even for
longer diffusion times.}
\end{figure}


\subsection{Convergence}
\label{subsec:exp-convergence}

As we have shown in \Cref{subsubsec:convergence} the estimator converges to
its expectation at a rate of $O(n^{-1/2})$ w.r.t.\ the RMSE. In 
\Cref{subsubsec:low-discrepancy} we introduced the idea of using 
low-discrepancy sequences. Theoretically, they should lead to much faster 
convergence, thus here we test whether this also holds in practice.
In the experiments we again use the $256\times 256$ test image \emph{peppers}.
We use two sampling strategies for the masks, whose sample
means $\bm{c} = \frac{1}{n}\sum_{\ell=1}^{n}\bm{c}^\ell$ converge to the
same expectation $\mathbb{E}[\bm{c}|\bm{f}]$.
As a representative of a low-discrepancy sequence we use the R2 
sequence~\cite{Ro18} to create a sampling threshold in each pixel 
(see \cite{Ro18} for details). 
This leads to a more regular sampling pattern compared to using a 
purely random threshold.
To make the experiment relevant to realistic scenarios, 
we use the analytic mask selection method from \Cref{subsec:analytic}.
We first test the mask convergence. 
To this end, we create $2^{16} = 65536$ masks via Poisson sampling and 
consider their average as converged to the expectation 
$\mathbb{E}[\bm{c}|\bm{f}]$.
Then we sample masks with both sampling strategies and observe how the 
RMSE between sample mean and expectation evolves with $n$.
Of course, we are more interested in the convergence of the
DbI result $\langle\bm{u}\rangle_n$.
Therefore, following a similar approach as for the masks, we create an 
individual ``converged'' DbI result for the two sampling methods, and 
again consider the RMSE between $\langle\bm{u}\rangle_n$ and the 
respective reference images.
\Cref{fig:exp-convergence} shows that the simple Poisson sampling leads to 
a convergence rate of $O(n^{-1/2})$ for the masks 
as well as for the DbI result,
which is perfectly in line with the theory from \Cref{subsubsec:convergence}.
Through low-discrepancy sampling this rate approaches $O(n^{-1})$. 
By fitting a curve through the data, we get a convergence rate of
$O(n^{-0.77})$ for the masks and $O(n^{-0.78})$ for the DbI result. 
The $O(n^{-1})$ estimate is typically achieved for low dimensions, 
so the difference of our results can be explained by the high 
dimensionality of our sampling problem.
The experiments confirm that the sampling strategy based on 
low-discrepancy sequences is indeed able to improve the convergence in 
practice.


\begin{figure}[tbhp]
\centering
\tabcolsep1pt
\begin{center}
\begin{tabular}{cc}
\textbf{(a) Mask Convergence} &
\textbf{(b) DbI Convergence} \\[1.2mm]
\includegraphics[width=0.49\textwidth]
{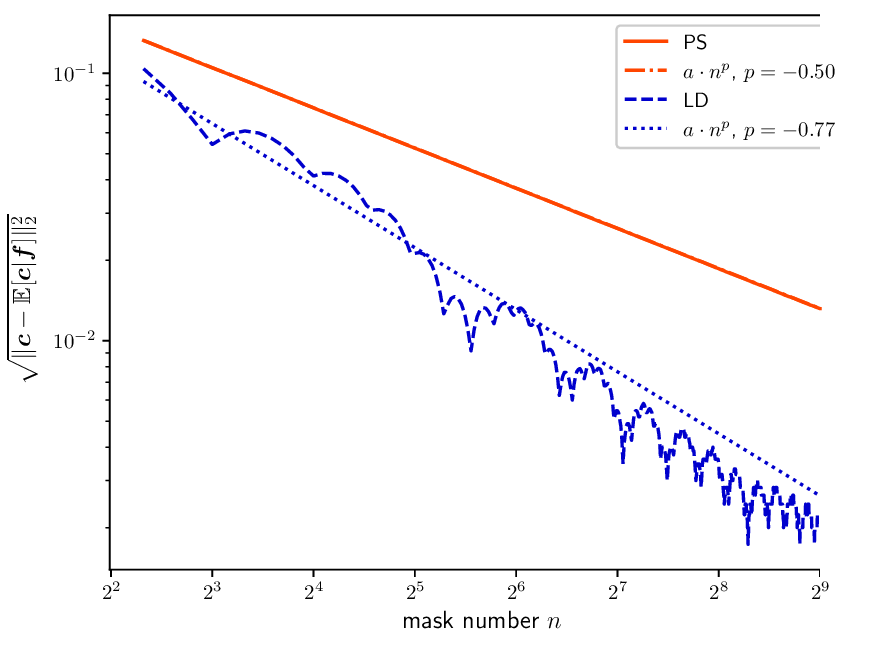} &
\includegraphics[width=0.49\textwidth]
{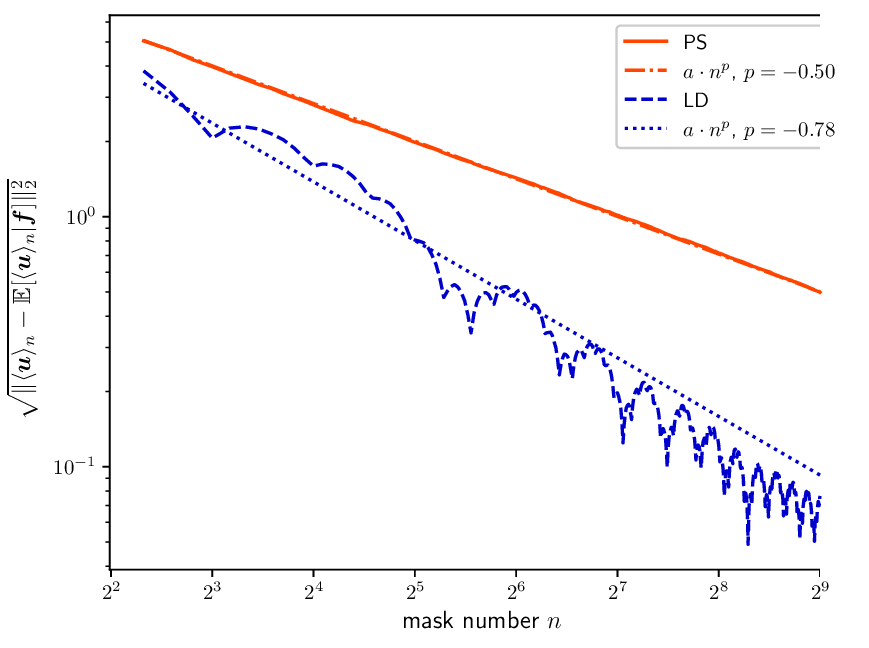}
\end{tabular}
\end{center}
\caption{\label{fig:exp-convergence} Convergence results for denoising by 
inpainting with the analytic method with Poisson sampling (PS) vs. 
low-discrepancy-based sampling (LD). (a) shows the convergence of the masks
and (b) the convergence of the DbI result.} 
\end{figure}


\subsection{Data Optimization for Denoising by Inpainting}
\label{subsec:exp-mask-optimization}

In the next step, we investigate the edge-preserving filtering behavior 
achieved by the use of adaptive masks. 
We first test the two spatial optimization methods and compare the results
to classical diffusion models. We show that DbI can yield results comparable 
to certain space-variant diffusion methods.
Then we discuss the effect of tonal optimization in the DbI setting.
It should be noted that these experiments are meant to provide an 
illustration of the mask optimization strategies and not to achieve 
the best denoising
quality. As we have shown, these strategies can be applied in a more general
setting than DbI with homogeneous diffusion inpainting. 
They are valid for the general probabilistic framework from 
\Cref{subsec:prob-theory}, and as such they also extend to more complex 
operators (including nonlinear ones).

We perform experiments on the three standard test images 
\emph{trui}, \emph{peppers} and \emph{walter} with a resolution of 
$256\times 256$, that are corrupted with additive Gaussian noise with standard 
deviations $\sigma_n \in \{10, 20, 30\}$ that we do not clip.
To ensure a fair comparison, we optimize the mask density and if required the 
pre- and post-smoothing parameter for the denoising by inpainting methods 
w.r.t.\ the MSE to the original image. We do this individually for each image 
and for each noise level using a grid search. In practice, these parameters 
need to be adapted to the noise level and the image content. We create 32 
masks with each of the mask selection methods, except for the regular masks 
where the number is determined by the spacing and thus by the density. 
For the proposed probabilistic densification algorithm we set the number of 
candidate pixels per iteration to 16.


\subsubsection{Spatial Optimization}
\label{subsubsec:exp-spatial-optimization}

Firstly, we investigate the different spatial selection strategies 
proposed in \Cref{sec:adaptive} and compare the denoising results 
with the standard diffusion methods presented in \Cref{subsec:denoising}. 
For the diffusion methods, which we discretize with an explicit scheme,
we optimize the stopping time and if required the contrast parameter 
of the Charbonnier diffusivity~\cite{CBAB97}.

As can be seen in \Cref{tab:spatial}, inpainting with regular masks leads 
to unsatisfying results, slightly worse than those obtained with 
homogeneous diffusion filtering. This is expected given the connections 
derived in \Cref{subsec:regular}. Note that the stopping time in homogeneous 
diffusion filtering can be tuned continuously, while the spacing of the 
regular mask can only be adapted in integer steps.
The analytic method based on Poisson sampling of the smoothed Laplacian 
magnitude improves the results, especially at lower noise levels. 
\Cref{fig:spatial}(c) shows how the mask pixels accumulate around important 
image structures, enabling an edge-preserving filtering behavior.
The densification method is able to further improve those results. The 
reason for this improvement can be seen in \Cref{fig:spatial}(d). 
On top of selecting pixels at reasonable positions, the error in the mask 
is reduced drastically in comparison to the analytic method, because 
densification implicitly avoids noisy pixels.
The adaptive mask selection strategies enable the denoising by inpainting 
method to produce results that are comparable to linear space-variant 
diffusion filtering. However, it cannot reach the 
quality of nonlinear diffusion. This is not surprising, as a feedback 
mechanism throughout the inpainting process is missing.
Nonetheless, the results reveal that proper data optimization enables DbI 
to compete with methods that optimize the operator, if they are of
comparable complexity.

Although qualitatively the densification approach is better than the analytic
method, its required run time is orders of magnitude larger, and this only 
gets worse for images of higher resolution. Due to the required number of 
inpaintings, the densification method takes about an hour to create a single 
mask with 10\% density for our $256\times 256$ pixel test images. 
In contrast, the analytic and the regular approaches allow instant mask 
generation in approximately a millisecond.
Thus, the analytic method yields a reasonable spatial mask pixel distribution 
in a very short time and clearly has potential, if the error in the mask 
pixels can be reduced. 
We show next that this can be achieved by complementing the mask selection 
strategies with tonal optimization.


\begin{table}[tbhp]
\footnotesize
\caption{Results (MSE) for denoising by inpainting with regular masks, 
the densification method and the analytic method 
with 32 masks (fewer masks for the regular mask method).
Comparison to classical diffusion-based denoising methods.
}\label{tab:spatial}
\centering
\begin{tabular}{c|c|c|c|c|c|c|c|c|c|c}
    & & \multicolumn{3}{c|}{\emph{trui}} & 
    \multicolumn{3}{c|}{\emph{peppers}} & 
    \multicolumn{3}{c}{\emph{walter}} \\
    & noise level $\sigma_n$ 
    & $10$ & $20$ & $30$ 
    & $10$ & $20$ & $30$ 
    & $10$ & $20$ & $30$ \\ \hline
    \multirow{3}{*}{\rotatebox[origin=c]{90}{DbI}} 
    & regular & $27.30$ & $57.29$ & $86.46$ 
    & $35.31$ & $64.40$ & $91.79$ 
    & $22.63$ & $50.13$ & $79.16$ \\
    & densification & $\bm{19.34}$ & $\bm{42.72}$ & $\bm{68.01}$ 
    & $\bm{24.36}$ & $\bm{47.27}$ & $\bm{69.89}$ 
    & $\bm{13.40}$ & $\bm{29.65}$ & $\bm{47.65}$ \\
    & analytic & $21.49$ & $49.71$ & $79.79$ 
    & $25.14$ & $51.70$ & $79.91$ 
    & $16.41$ & $37.83$ & $62.08$ \\ \hline
    \multirow{3}{*}{\rotatebox[origin=c]{90}{Diff}} 
    & homogeneous & $24.12$ & $50.18$ & $76.12$ 
    & $32.16$ & $59.77$ & $84.58$ 
    & $19.65$ & $42.76$ & $66.87$ \\
    & lin. space-var. & $17.89$ & $42.62$ & $69.57$ 
    & $24.03$ & $47.47$ & $72.67$ 
    & $13.31$ & $32.30$ & $55.37$ \\
    & nonlinear & $\bm{16.21}$ & $\bm{34.99}$ & $\bm{54.66}$ 
    & $\bm{22.63}$ & $\bm{40.48}$ & $\bm{57.54}$ 
    & $\bm{11.89}$ & $\bm{25.31}$ & $\bm{39.49}$ \\ 
\end{tabular}
\end{table}


\begin{figure}[tbhp]
\centering
\tabcolsep5pt
\begin{center}
\begin{tabular}{cccc}
\textbf{(a) input} & \textbf{(b) regular} & 
\textbf{(c) densification} & \textbf{(d) analytic} \\[1.2mm]
\includegraphics[width=0.22\textwidth]
{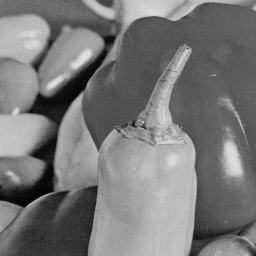} &
\includegraphics[width=0.22\textwidth]
{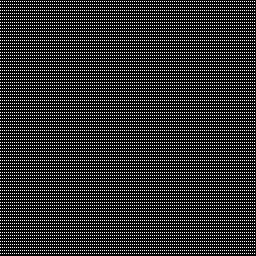} &
\includegraphics[width=0.22\textwidth]
{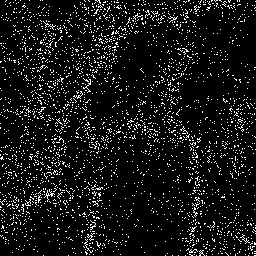}  &
\includegraphics[width=0.22\textwidth]
{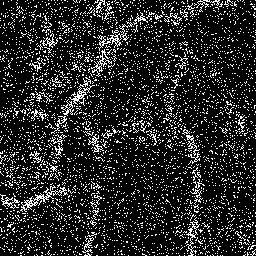} \\[-0.8mm]
original & Mask MSE: $405.09$ & 
Mask MSE: $345.23$  & Mask MSE: $410.63$ \\[1.2mm]
\includegraphics[width=0.22\textwidth]
{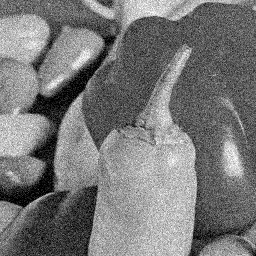} &
\includegraphics[width=0.22\textwidth]
{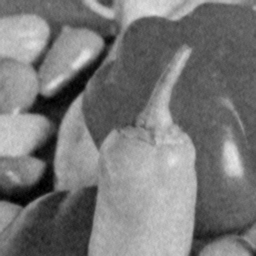} &
\includegraphics[width=0.22\textwidth]
{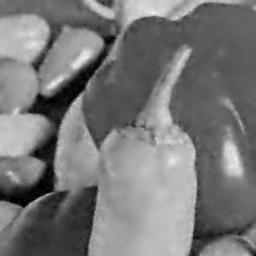}  &
\includegraphics[width=0.22\textwidth]
{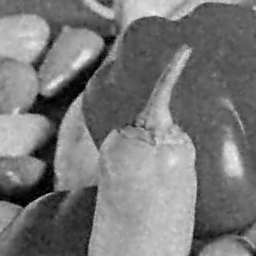} \\[-0.8mm]
noisy ($\sigma_n=20$) & MSE: $64.40$ & MSE: $47.27$ & MSE: $51.70$
\end{tabular}
\end{center}
\caption{\label{fig:spatial}
Results for denoising by inpainting with 32 masks 
(six masks for the regular mask method) for the different spatial 
optimization methods on the test image \emph{peppers} with $\sigma_n=20$. 
Top row: (a) original image, (b)-(d) one representative out of all 
the masks for every method. The MSE is computed at mask pixels. 
Bottom row: (a) noisy image, (b)-(d) denoising by inpainting results with 
optimized parameters and the MSE in the entire image. 
We see that our analytic method and the densification method adapt the 
mask point locations to the structure of the image. 
Densification additionally avoids choosing noisy mask pixels, leading to 
a smaller error in the mask pixels and eventually to a better 
reconstruction.}
\end{figure}


\subsubsection{Tonal Optimization}
\label{subsubsec:exp-tonal-optimization}

As mentioned in \Cref{ss:MMSE_and_Tonal_Optimization}, tonal 
optimization leads to an MMSE estimate that is approximating 
instead of interpolating. 
If one assumes that mask pixels are erroneous due to the noise, this is
certainly a desirable behavior. We will evaluate its effect in the following. 
To this end, we apply tonal optimization to the masks obtained by 
each of our spatial optimization methods. We optimize the tonal values for 
each individual mask, before once again averaging the respective inpaintings 
to obtain the final denoised result.

The results in \Cref{tab:tonal} reveal that the methods that do not consider 
the noise in the selection process get the greatest boost in performance. 
This confirms the conjecture that tonal optimization is able to mitigate the 
negative effect of noisy mask pixels selection.
We also observe that tonal optimization decreases the error in the mask pixels 
for those methods. In \Cref{fig:spatial}, the MSE at mask locations decreases 
from 405.09 to 271.80 for the regular mask, and from 410.63 to 306.62 for 
the analytic method.
For probabilistic densification, tonal optimization barely changes the final 
results, as well as the mask MSE (which even increases slightly from 345.23 
to 356.12 in the example).

We see that tonal optimization enables the analytic method to produce 
results of quality comparable to those of the densification method, 
and of better quality than space-variant diffusion.
Although the tonal optimization step takes some additional seconds, 
the analytic method is still orders of magnitude faster than the 
densification method.
\Cref{fig:tonal-results} shows a selection of resulting images comparing 
the two adaptive mask selection methods with tonal optimization 
and linear space-variant diffusion, as the diffusion method 
that leads to the most similar results.


\begin{table}[tbhp]
\footnotesize
\caption{\label{tab:tonal} Results (MSE) for denoising by inpainting 
with regular masks, the densification method and the analytic method 
with 32 masks 
(less masks for the regular mask method) including tonal optimization.
Comparison to classical diffusion-based denoising methods.}
\centering
\begin{tabular}{c|c|c|c|c|c|c|c|c|c|c}
    & & \multicolumn{3}{c|}{\emph{trui}} & 
    \multicolumn{3}{c|}{\emph{peppers}} & 
    \multicolumn{3}{c}{\emph{walter}} \\
    & noise level $\sigma_n$ 
    & $10$ & $20$ & $30$ 
    & $10$ & $20$ & $30$ 
    & $10$ & $20$ & $30$ \\ \hline
    \multirow{3}{*}{\rotatebox[origin=c]{90}{DbI}} 
    & regular & $22.32$ & $48.77$ & $76.06$ 
    & $32.65$ & $60.92$ & $87.05$ 
    & $16.36$ & $39.54$ & $64.45$ \\
    & densification & $18.46$ & $41.56$ & $67.72$ 
    & $24.42$ & $47.28$ & $70.21$ 
    & $12.35$ & $28.13$ & $45.92$ \\
    & analytic & $\bm{17.24}$ & $\bm{39.49}$ & $\bm{63.17}$ 
    & $\bm{23.68}$ & $\bm{46.43}$ & $\bm{68.55}$ 
    & $\bm{12.08}$ & $\bm{27.66}$ & $\bm{45.36}$ \\ \hline
    \multirow{3}{*}{\rotatebox[origin=c]{90}{Diff}} 
    & homogeneous & $24.12$ & $50.18$ & $76.12$ 
    & $32.16$ & $59.77$ & $84.58$ 
    & $19.65$ & $42.76$ & $66.87$ \\
    & lin. space-var. & $17.89$ & $42.62$ & $69.57$ 
    & $24.03$ & $47.47$ & $72.67$ 
    & $13.31$ & $32.30$ & $55.37$ \\
    & nonlinear & $\bm{16.21}$ & $\bm{34.99}$ & $\bm{54.66}$ 
    & $\bm{22.63}$ & $\bm{40.48}$ & $\bm{57.54}$ 
    & $\bm{11.89}$ & $\bm{25.31}$ & $\bm{39.49}$ \\
\end{tabular}
\end{table}


\begin{figure}[tbhp]
\centering
\tabcolsep5pt
\begin{center}
\begin{tabular}{cccc}
\textbf{(a) noisy input} & \textbf{(b) lin. sp.-var.} & 
\textbf{(c) densification} & \textbf{(d) analytic} \\[1.2mm]
\includegraphics[width=0.22\textwidth]
{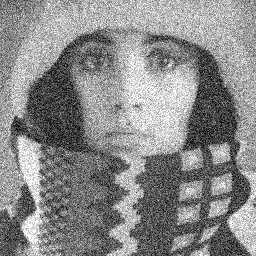} &
\includegraphics[width=0.22\textwidth]
{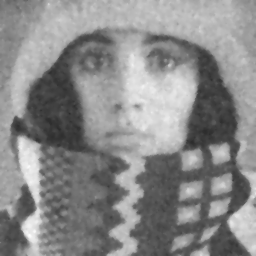} &
\includegraphics[width=0.22\textwidth]
{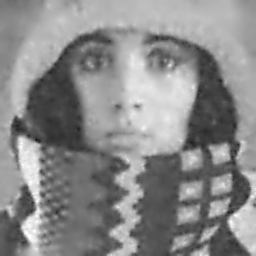}  &
\includegraphics[width=0.22\textwidth]
{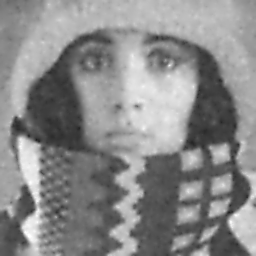} \\[-0.8mm]
\emph{trui}, $\sigma_n=30$ & MSE: $69.57$ & MSE: $67.72$ & MSE: $63.17$ \\[1.2mm]
\includegraphics[width=0.22\textwidth]
{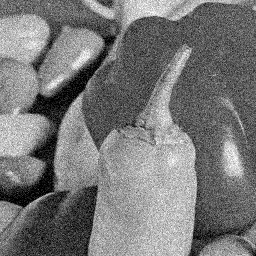} &
\includegraphics[width=0.22\textwidth]
{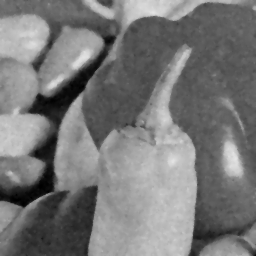} &
\includegraphics[width=0.22\textwidth]
{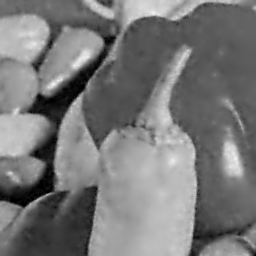}  &
\includegraphics[width=0.22\textwidth]
{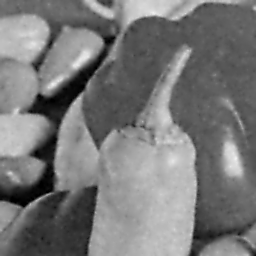} \\[-0.8mm]
\emph{peppers}, $\sigma_n=20$ & MSE: $47.47$ & MSE: $47.28$ & MSE: $46.43$ \\[1.2mm]
\includegraphics[width=0.22\textwidth]
{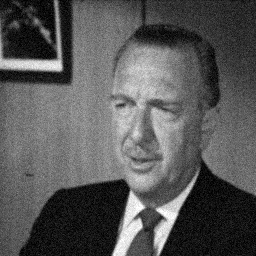} &
\includegraphics[width=0.22\textwidth]
{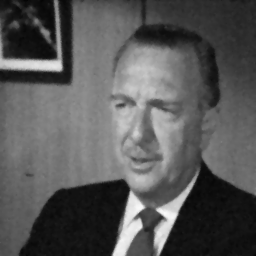} &
\includegraphics[width=0.22\textwidth]
{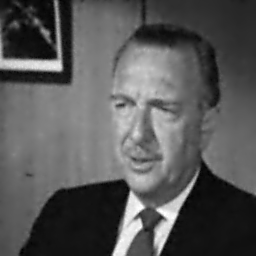} &
\includegraphics[width=0.22\textwidth]
{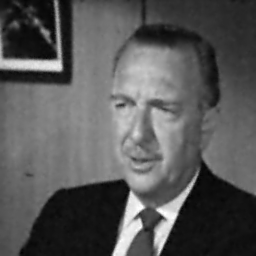} \\[-0.8mm]
\emph{walter}, $\sigma_n=10$ & MSE: $13.31$& MSE: $12.35$ & MSE: $12.08$ 
\end{tabular}
\end{center}
\caption{\label{fig:tonal-results} Visual comparison of linear space-variant
diffusion and denoising by inpainting with the densification method and 
the analytic method on three test images with noise. 
Both DbI methods are using tonal optimization.}
\end{figure}


\subsection{Denoising by Biharmonic Inpainting}
\label{subsec:exp-biharmonic}

Our previous results reveal that optimizing the data instead of the operator
constitutes an interesting alternative for image denoising. 
To further substantiate this idea, we now adapt the inpainting operator 
within the DbI framework. We consider biharmonic inpainting as a 
representative of a higher-order polyharmonic operator.

It has been shown that the biharmonic operator can have quality 
advantages over homogeneous diffusion (i.e., the harmonic operator) 
in classical sparse inpainting~\cite{CRP14, GWWB08, SPMEWB14}.
Biharmonic inpainting is given by the PDE
\begin{equation}
\label{eq:biharmonic-inp-pde-single}
\bigl(c(\bm{x}) + (1 - c(\bm{x})) \Delta^2 \bigr)u(\bm{x}) 
= c(\bm{x})f(\bm{x}) \quad \text{for } \bm{x} \in \Omega,
\end{equation}
with 
$\Delta^2 u = \partial_{xxxx} u + 2 \partial_{xxyy} u + \partial_{yyyy} u$ 
and reflecting boundary conditions $\partial_{\bm{n}}u(\bm{x}) = 0$ and 
$\partial_{\bm{n}}\Delta u (\bm{x}) = 0$ for $\bm{x}\in\partial\Omega$. 
It can be derived from the following variational formulation 
(analogously to \labelcref{eq:harmonic_variational_formulation}):
\begin{equation}
\label{eq:biharmonic_variational_formulation}
\min_{u}\int_{\Omega}(\Delta u(\bm{x}))^2\,d\bm{x}, 
\textrm{ such that } u(\bm{x}) = f(\bm{x}) \textrm{ for } \bm{x} \in K.
\end{equation}
This shows that biharmonic inpainting penalizes second-order derivatives.
Biharmonic inpainting does not suffer from the typical singularities at mask 
points that homogeneous diffusion inpainting produces. On the other hand it 
can produce over- and undershoots, since it does not guarantee a 
maximum-minimum principle.
We evaluate the potential of biharmonic inpainting for denoising by comparing 
it to homogeneous diffusion inpainting. 
To ensure that the results reflect the quality of the operators, 
we first perform the experiment on fully random masks.

Our results in \Cref{tab:biharmonic-random} show that biharmonic inpainting 
does lead to an improvement, and it is largest at low noise levels. 
This is to be expected, as the method is not as radical as homogeneous 
diffusion inpainting, since it penalizes second degree instead of first 
degree derivatives. 
However, already tonal optimization as a first data optimization step 
neutralizes this advantage and the two methods perform similarly. 
These results support our reasoning that data optimization plays a 
significant role for the denoising abilities of our framework, being 
more important than the use of more complex, higher-order models. 
Further experiments on spatially optimized masks 
(see \Cref{tab:biharmonic-analytic}) confirm our findings, and 
even shift the advantage towards homogeneous diffusion inpainting. 
When comparing to previous results from classical sparse image inpainting, 
one has to consider that the singularities, 
that homogeneous diffusion inpainting suffers from, are suppressed by the 
averaging in the DbI framework. Thus, this disadvantage of homogeneous 
diffusion inpainting does not come into play in our scenario. 
Lastly, one should keep in mind that biharmonic inpainting leads to a higher 
condition number of the inpainting matrix, and consequently each inpainting 
is numerically more burdensome and less efficient.


\begin{table}[tbhp]
\footnotesize
\caption{Results (MSE) for denoising by inpainting with 32 random masks 
using homogeneous diffusion (HD) and biharmonic (BI) inpainting, without 
and with tonal optimization (TO).}\label{tab:biharmonic-random}
\centering
\begin{tabular}{c|c|c|c|c|c|c|c|c|c}
    & \multicolumn{3}{c|}{\emph{trui}} & 
    \multicolumn{3}{c|}{\emph{peppers}} & 
    \multicolumn{3}{c}{\emph{walter}} \\
    noise level $\sigma_n$ 
    & $10$ & $20$ & $30$ 
    & $10$ & $20$ & $30$ 
    & $10$ & $20$ & $30$ \\ \hline
    HD, without TO & $30.53$ & $65.51$ & $100.51$ 
    & $36.01$ & $71.20$ & $104.21$ 
    & $26.93$ & $60.21$ & $94.37$  \\
    BI, without TO & $\bm{24.23}$ & $\bm{56.37}$ & $\bm{93.51}$ 
    & $\bm{33.28}$ & $\bm{66.61}$ & $\bm{102.49}$ 
    & $\bm{19.16}$ & $\bm{48.19}$ & $\bm{82.92}$  \\
    \hline
    HD, with TO & $23.10$ & $49.83$ & $\bm{76.57}$ 
    & $\bm{31.98}$ & $\bm{59.75}$ & $\bm{85.28}$ 
    & $18.09$ & $41.79$ & $66.40$ \\
    BI, with TO & $\bm{22.21}$ & $\bm{49.74}$ & $77.25$ 
    & $33.27$ & $61.84$ & $87.92$ 
    & $\bm{16.52}$ & $\bm{39.53}$ & $\bm{65.03}$ \\
\end{tabular}
\end{table}


\begin{table}[tbhp]
\footnotesize
\caption{Results (MSE) for denoising by inpainting with 32 masks obtained with 
the analytic method using homogeneous diffusion (HD) and biharmonic (BI) 
inpainting, without and with tonal optimization (TO).}
\label{tab:biharmonic-analytic}
\centering
\begin{tabular}{c|c|c|c|c|c|c|c|c|c}
    & \multicolumn{3}{c|}{\emph{trui}} & 
    \multicolumn{3}{c|}{\emph{peppers}} & 
    \multicolumn{3}{c}{\emph{walter}} \\
    noise level $\sigma_n$ 
    & $10$ & $20$ & $30$ 
    & $10$ & $20$ & $30$ 
    & $10$ & $20$ & $30$ \\ \hline
    HD, without TO & $21.49$ & $49.71$ & $\bm{79.79}$ 
    & $\bm{25.14}$ & $\bm{51.70}$ & $\bm{79.91}$ 
    & $16.41$ & $37.83$ & $\bm{62.08}$  \\
    BI, without TO & $\bm{19.01}$ & $\bm{47.47}$ & $82.39$ 
    & $25.83$ & $55.42$ & $90.28$ 
    & $\bm{14.16}$ & $\bm{37.15}$ & $68.25$  \\
    \hline
    HD, with TO & $17.24$ & $\bm{39.49}$ & $\bm{63.17}$ 
    & $\bm{23.68}$ & $\bm{46.43}$ & $\bm{68.55}$ 
    & $12.08$ & $27.66$ & $\bm{45.36}$  \\
    BI, with TO & $\bm{17.18}$ & $40.45$ & $66.13$ 
    & $25.35$ & $49.27$ & $72.68$ 
    & $\bm{11.74}$ & $\bm{27.22}$ & $45.70$ \\
\end{tabular}
\end{table}


\section{Conclusions}
\label{sec:conclusion}

Our work is the first that links the tasks of PDE-based image inpainting and 
denoising in a systematic way, by providing an explicit connection
between homogeneous diffusion inpainting and denoising through a relation 
between the diffusion time and the mask density.
Our {\em denoising by inpainting (DbI)} framework 
achieves denoising by averaging inpainting results with different sparse 
masks of the same density. It consitutes a means to investigate the 
connections between PDE-based denoising and inpainting and allows us to
evaluate the denoising potential of PDE-based inpainting methods.
We have established a probabilistic theory with convergence estimates for 
the framework, and have extended it to a deterministic version by the use of
low-discrepancy sequences. 
We have further shown that this framework computes an approximation to an 
MMSE estimate.
For non-adaptive masks we have linked the framework to classical diffusion 
via a one-to-one relationships between the mask density and the diffusion 
time. 
We have demonstrated that a simple operator can exhibit space-variant 
filtering behavior, 
when supplemented with adaptive data selection strategies. 
Experiments with a higher-order inpainting operator, which can be more 
powerful than homogeneous diffusion inpainting~\cite{CRP14, GWWB08,SPMEWB14}, 
have underlined the importance of choosing appropriate 
data over more complex operators.
For data optimization specific to denoising by inpainting, we have presented 
two distinct, fundamental strategies. The densification method from our 
conference paper~\cite{APW17} aims at finding pixels that represent the 
data well. Thereby, it implicitly avoids the selection of noisy mask pixels 
during spatial optimization.
On the contrary, we have proposed a new approach, where the selection of 
noisy pixels is tolerated in the spatial optimization but is compensated
for by the tonal optimization.

\medskip
Our work constitutes an unconventional, new viewpoint on image denoising:
By using a simple inpainting operator but focusing on adequate data 
selection we \emph{shift the priority from optimizing the filter 
model to optimizing the considered data.}
Moreover, our densification strategy allows us to find the most 
trustworthy pixels in the data. 
This shows that \emph{simple filter operators such 
as homogeneous diffusion can give deep insights into data}.
Last but not least, we have seen that the filling-in effect is not only
useful in variational optic flow models and in PDE-based inpainting,
but also in denoising. This emphasizes its fundamental role in digital
image analysis, which is in full agreement with classical results 
from biological vision~\cite{We35}.

\medskip
While our focus in the present paper is on gaining fundamental insights
into the potential of inpainting ideas for denoising, our future work
will deal with various modifications to make these ideas also applicable 
to more recent denoising methods.
To this end, we are going to consider more sophisticated inpainting
operators~\cite{SPMEWB14} and data selection strategies~\cite{DAW21}, 
including neural ones~\cite{PSAW23}, and the incorporation of 
more advanced types of data~\cite{JCW23}. 
Such future work should also extend our theory to, e.g., space-variant 
and nonlinear operators.


\backmatter

\section*{Declarations}

\paragraph{Acknowledgements}
Not applicable.

\paragraph{Author contribution}
DG conducted experiments and provided theory. VC provided mathematical 
theory. DG and VC wrote the paper, with feedback from JW and PP.
RDA, PP and JW jointly contributed to the conference paper.
The authors have read and approved the final manuscript.

\paragraph{Data/Code availability}
The datasets used and/or analyzed during the current study are available 
from the corresponding author on reasonable request.

\paragraph{Funding}
This project has received funding from the European Research 
Council (ERC) under the European Union's Horizon 2020 research and 
innovation programme (grant agreement No 741215, ERC Advanced Grant 
INCOVID).

\paragraph{Competing interests}
The authors declare that they have no competing interests.


\section*{List of Abbreviations}

\begin{description}[leftmargin=!, labelwidth=\widthof{\textit{MMSE}}]
    \item[DbI] Denoising by Inpainting
    \item[DCT] Discrete Cosine Transform
    \item[EED] Edge-Enhancing Diffusion
    \item[MAP] Maximum A Posteriori
    \item[ML] Maximum Likelihood
    \item[MMSE] Minimum Mean Squared Error
    \item[MSE] Mean Squared Error
    \item[PDE] Partial Differential Equation
    \item[PDF] Probability Density Function
    \item[PMF] Probability Mass Function
    \item[RMSE] Root Mean Square Error
\end{description}


\begin{appendices}


\section{Proof of \Cref{thm:pmf-dens}}
\label[appendix]{app:densification}

We derive the expression for the stated probability in \Cref{thm:pmf-dens} 
for step $k+1$ here.
At the beginning of step $k+1$, \Cref{alg:densification} has already inserted 
$k$ mask pixels yielding the mask $\bm{c}^k$. At the end of step $k+1$ we want 
to have inserted a new mask pixel that is not in $\bm{c}^k$. Consequently, 
we can select a pixel only from the set 
$\mathcal{I}^k$ of remaining empty mask pixel locations, 
with $|\mathcal{I}^k|=N-k$.
The algorithm samples a set $\mathcal{X}$ of $\alpha$ distinct candidates 
from $\mathcal{I}^k$ uniformly at random (there are 
$C^{N-k}_{\alpha}$ different ways to do so):
\begin{equation}
\mathcal{X} = 
\{X_1, \ldots, X_\alpha \in \mathcal{I}^k \,:\, X_i\ne X_j 
\text{ for } i\ne j\}.
\end{equation}
Then one chooses the candidate $X^*\in\mathcal{X}$ with lowest reconstruction 
error (w.r.t.\ the noisy image $\bm{f}$):
\begin{equation}
X^*\in \mathcal{X}^* = \underset{X\in\mathcal{X}}{\argmin}\,E^k(X), 
\quad E^k(X) \coloneqq \|\bm{r}(\bm{c}^k+\bm{e}_X, \bm{f})-\bm{f}\|^2_2,
\end{equation}
where $\bm{e}_X\in\R^N$ is the zero vector modified with a one at the location 
corresponding to mask point $X$. The minimizer does not have to be unique; 
in fact the set of minimizers
\begin{equation}
\mathcal{X}^* = \{X_i\in\mathcal{X}\,:\,E^k(X_i) 
= \min_{X\in\mathcal{X}}E^k(X)\}
\end{equation}
may have more than one element ($|\mathcal{X}^*|>1$) in which case we choose 
$X^*$ uniformly at random from $\mathcal{X}^*$ with probability 
$\frac{1}{|\mathcal{X}^*|}$. This completes step $k+1$, now with a specific 
$X^* = x^*$ and corresponding mask $\bm{c}^{k+1} = \bm{c}^k + \bm{e}_{x^*}$. 
If the desired number of mask points have been achieved the algorithm ends, 
otherwise one proceeds to step $k+2$ in the exact same manner.

After we have inserted mask pixel $x^*\in\mathcal{I}^k$ we want to be able 
to compute the probability $\Pr(X^* = x^*)$ of this occurring. This is equal 
to the probability of $x^*$ having been selected as a candidate:
\begin{equation}
\Pr(x^*\in\mathcal{X}) 
= C^{N-k-1}_{\alpha-1}/C^{N-k}_{\alpha} = \frac{\alpha}{N-k},
\end{equation}
multiplied by the probability 
$\Pr(x^*\in\mathcal{X}^*\,|\,x^*\in\mathcal{X})$ that $x^*$ ends up in 
$\mathcal{X}^*$, which is in turn multiplied by the probability 
$\Pr(X^*=x^*\,|\,x^*\in\mathcal{X}^*) = \frac{1}{|\mathcal{X}^*|}$ of having 
picked $x^*$ from $\mathcal{X}^*$ uniformly at random. We thus have the 
following chain of conditional probabilities:
\begin{align}
\begin{split}
\Pr(X^*=x^*) 
&= \Pr(X^*=x^*\,|\,x^*\in\mathcal{X}^*)\Pr(x^*\in\mathcal{X}^*) \\
&= \Pr(X^*=x^*\,|\,x^*\in\mathcal{X}^*)
\Pr(x^*\in\mathcal{X}^*\,|\,x^*\in\mathcal{X})\Pr(x^*\in\mathcal{X}) \\
&= \frac{1}{|\mathcal{X}^*|}\Pr(x^*\in\mathcal{X}^*\,|\,x^*\in\mathcal{X})
\frac{\alpha}{N-k}.
\end{split}
\end{align}
We can write the terms involving $\mathcal{X}^*$ in the following manner:
\begin{equation}
\frac{1}{|\mathcal{X}^*|}\Pr(x^*\in\mathcal{X}^*\,|\,x^*\in\mathcal{X}) = 
\sum_{\beta=1}^{\alpha}\frac{1}{\beta}
\Pr(x^*\in\mathcal{X}^*\land |\mathcal{X}^*|=\beta\,|\,x^*\in\mathcal{X}).
\end{equation}
The probability on the right-hand side can be rewritten as requiring $\beta$ 
of the candidates to have energy equal to $E^k(x^*)$ and the remaining 
$\alpha-\beta$ having a strictly larger energy:
\begin{equation}
\begin{gathered}
\Pr(x^*\in\mathcal{X}^*\land |\mathcal{X}^*|= \beta\,|\,x^*\in\mathcal{X})
= \\ = \Pr\left(\left(\bigwedge_{i=1}^{\beta}E^k(X_i)=E^k(x^*)\right)
\land\left( \bigwedge_{j=\beta+1}^{\alpha} E^k(X_j)>E^k(x^*)\right)\,
\Bigg|\,x^*\in\mathcal{X}\right).
\end{gathered}
\end{equation}
To compute the above probabilities we would need to know the total number 
of pixels from $\mathcal{I}^k$ with energy equal to $E^k(x^*)$:
\begin{equation}
N_{eq} \coloneqq |\{x\in\mathcal{I}^k\,:\,E^k(x)= E^k(x^*)\}|,
\end{equation}
and the total number of pixels from $\mathcal{I}^k$ having a strictly 
higher energy:
\begin{equation}
N_{gt} \coloneqq |\{x\in\mathcal{I}^k\,:\,E^k(x)> E^k(x^*)\}|.
\end{equation}
From the requirement $|\mathcal{X}^*|=\beta$, it follows that we need to 
choose $\beta$ pixels that have energy equal to $E^k(x^*)$. 
However, $x^*\in\mathcal{X}$ so $E^k(X)=E^k(x^*)$ with probability $1$ for 
at least one candidate $X=x^*$. Then $\beta-1$ elements $X_i$ remain to be 
selected from $N_{eq}-1$ locations, the total number of possibilities 
being $C^{N_{eq}-1}_{\beta-1}$. Finally the remaining $\alpha-\beta$ 
candidates must be selected from $N_{gt}$ locations, 
resulting in $C^{N_{gt}}_{\alpha-\beta}$ options. 
Using this we can compute the probability
\begin{equation}
\frac{1}{|\mathcal{X}^*|}\Pr(x^*\in\mathcal{X}^*\,|\,x^*\in\mathcal{X}) 
= \sum_{\beta=1}^{\alpha}\frac{1}{\beta}
\frac{C^{N_{eq}-1}_{\beta-1}C^{N_{gt}}_{\alpha-\beta}}
{C^{N-k-1}_{\alpha-1}}.
\end{equation}
Ultimately we get the following probability for step $k+1$:
\begin{equation}
\Pr(X^*=x^*) = \frac{\alpha}{N-k}\sum_{\beta=1}^{\alpha}\frac{1}{\beta}
\frac{C^{N_{eq}-1}_{\beta-1}C^{N_{gt}}_{\alpha-\beta}}
{C^{N-k-1}_{\alpha-1}} 
= \sum_{\beta=1}^{\alpha}\frac{1}{\beta}
\frac{C^{N_{eq}-1}_{\beta-1}C^{N_{gt}}_{\alpha-\beta}}{C^{N-k}_{\alpha}}.
\end{equation}
Through the probabilistic densification procedure the exact same mask $\bm{c}$, 
with $\|\bm{c}\|_0$ mask pixels, can be constructed in $\|\bm{c}\|_0!$ 
different ways 
(the same set of mask pixels being introduced in all possible orders). 
That is, we get the probability mass function $p_{\sigma}(\bm{c}|\bm{f})$ 
over masks that also retain 
the order of insertion of their mask pixels (e.g., we can modify $\bm{c}$ by 
setting entries equal to one, to be equal to $k$: the step in which those 
were inserted). To get the usual probability mass function over binary masks 
we need to sum up the above probabilities over all $\|\bm{c}\|_0!$ permutations 
of point insertion orders. The main issue for practicality is 
that $N_{eq}$ and $N_{gt}$ must be known, which would require evaluating 
all possible $|\mathcal{I}^k|=N-k$ inpaintings for a single step. 
Nevertheless, Monte Carlo can be used to estimate the probabilities.


\section{Probability for Error Diffusion Masks}
\label[appendix]{app:error-diff}

Error diffusion halftoning (e.g., Floyd-Steinberg dithering~\cite{FS76}) can 
be used to produce a binary mask $\bm{c}\in\{0,1\}^N$ from a continuous density 
image $\bm{d}\in[0,1]^N$. The process involves iterating over the image pixels 
(e.g., in serpentine order), binarizing a single pixel at a given step, and 
then diffusing the error arising from the binarization to the set of currently 
non-visited pixels. This results in a sequence of images 
$\bm{d}=\bm{d}^1, \bm{d}^2,\ldots, \bm{d}^{N+1}=\bm{c}$. 
The binarization happens according to a thresholding step, which usually reads:
\begin{equation}
    c_k = d^{k+1}_k = 
    \begin{cases}
        0 &\text{for } d^k_k < 0.5, \\
        1 &\text{for } d^k_k \geq 0.5.
    \end{cases}
\end{equation}
Since we want to get multiple 
masks stochastically, we randomize the process by sampling a uniform random 
number $v_k\in[0,1]$ for pixel $k$, and then perform thresholding: 
\begin{equation}
    c_k = d^{k+1}_k = 
    \begin{cases}
        0 &\text{for } d^k_k < v_k, \\
        1 &\text{for } d^k_k \geq v_k.
    \end{cases}
\end{equation}
Then the probability mass function for mask $\bm{c}$ constructed from density 
image $\bm{d}$ is 
\begin{equation}
p_{\bm{d}}(\bm{c}) = 
\frac{1}{P}\prod_{k=1}^N (d^{k}_k(\bm{c}))^{c_k}
(1-d^{k}_k(\bm{c}))^{1-c_k}, 
\quad P = \sum_{\bm{c}\in\{0,1\}^N}\prod_{k=1}^N 
(d^{k}_k(\bm{c}))^{c_k}(1-d^{k}_k(\bm{c}))^{1-c_k}.
\end{equation}
In the above $d^k_k(\bm{c})$ are assumed to be clamped to $[0,1]$.
Note that while this bears similarity to Poisson sampling 
(see \labelcref{eq:poisson_sampling}), the probability 
$d^k_k(\bm{c})$ is conditioned on the probabilities in the $k$ previous 
steps. Algorithmically it is trivial to compute the numerator of 
the probability during the error diffusion process. 


\end{appendices}


\bibliography{dbi-references}

\end{document}